\documentclass[pdflatex,sn-mathphys-num]{sn-jnl}

\usepackage{amsmath,amssymb,amsfonts}
\usepackage{mathtools}
\usepackage{graphicx}
\usepackage{hyperref}
\usepackage{bbm}
\usepackage{bm}
\usepackage{amsthm}
\usepackage{booktabs}
\usepackage{longtable}
\usepackage[title]{appendix}
\makeatletter
\let\orcidlogo\@undefined
\makeatother
\usepackage{orcidlink}

\theoremstyle{thmstyleone}
\newtheorem{theorem}{Theorem}
\newtheorem{corollary}[theorem]{Corollary}
\newtheorem{proposition}[theorem]{Proposition}
\theoremstyle{thmstyletwo}
\newtheorem{remark}{Remark}

\newcommand{\HH}{\mathbb{H}_+}
\newcommand{\Hp}[1]{H^{#1}_+}
\newcommand{\BB}{\mathcal{B}}
\newcommand{\KK}{\mathcal{K}}
\newcommand{\LL}{\mathcal{L}}
\newcommand{\GG}{\mathcal{G}}
\newcommand{\II}{\mathcal{I}}
\newcommand{\PP}{\mathcal{P}}
\newcommand{\opnorm}{\mathrm{op}}
\newcommand{\Proj}{P}
\newcommand{\Qproj}{Q}
\newcommand{\Ktilde}{\tilde{\KK}}

\newcommand{\TrOp}{\mathrm{Tr}}

\begin{document}

\title[Operator-valued Hardy spaces for quantum memory kernels]{Operator-Valued Hardy Spaces and Kramers--Kronig Relations for Non-Markovian Quantum Memory Kernels}

\author*[1,2]{\fnm{Kejun} \sur{Liu}\,\orcidlink{0000-0003-1547-9280}}\email{kjliu@suda.edu.cn}

\affil*[1]{\orgdiv{State Key Laboratory of Bioinspired Interface Material Science, Institute of Nano \& Functional Materials}, \orgname{Soochow University}, \orgaddress{\city{Suzhou}, \postcode{215123}, \country{China}}}

\affil*[2]{\orgdiv{School of Physical Science and Technology}, \orgname{Soochow University}, \orgaddress{\city{Suzhou}, \postcode{215006}, \country{China}}}

\abstract{Retarded support, upper-half-plane holomorphy, and Hardy boundary control are distinct properties of a memory kernel. We give sufficient conditions linking them for the Nakajima--Zwanzig (NZ) kernel of an open system with finite-dimensional system Liouville space. If the projected kernel has an absolutely continuous real-axis representation with density $w\in L^1\cap L^{p_0}$, no singular part, and $p_0>1$, its transform lies in the operator-valued Hardy class $H^p(\mathcal B)$ for $1<p\leq p_0$, yielding principal-value and once-subtracted Kramers--Kronig (KK) boundary formulas. Microscopic unitary evolution gives holomorphy for reduced-state transforms without that spectral hypothesis. A perturbative force-fit equation has a first-order quotient pole at a simple baseline state-transform zero $\zeta$ only when the first-order inhomogeneous numerator satisfies $\tilde{\mathcal I}^{(1)}(\zeta)\neq0$; matrix reconstructions additionally require adjugate non-cancellation. For rational reconstructions, clearing and no-cancellation hypotheses convert a rational kernel pole into a genuine upper-half-plane propagator pole, incompatible with a uniformly bounded completely positive trace-preserving (CPTP) family when the transforms agree on an open set. We formalize these implications in Lean~4, providing machine-checked verification from explicitly named classical inputs. A finite Jaynes--Cummings calculation provides a shifted, distributional KK check with a $0.024\%$ fixed-truncation residual.}


\maketitle

\section{Introduction}
\label{sec:intro}

The Nakajima--Zwanzig (NZ) equation replaces the eliminated bath by an operator-valued memory kernel~\cite{nakajima1958,zwanzig1960,mori1965}. This kernel is the object used in many frequency-domain reconstructions, but its analytic properties are often imported without checking the reduction that produced it. The Kramers--Kronig (KK) relations, long established for causal response functions in optics and condensed matter, have to our knowledge not been rigorously developed for the memory kernels of open quantum systems. This extension is nontrivial: a growing causal kernel need not have a transform on all of the upper half-plane $\HH$, and a singular or slowly decaying spectral weight may allow only a distributional or subtracted boundary relation. In this paper we bring the NZ memory kernel into the KK framework; without this step, one may miss obstructions that the causal structure places on reconstructions. In dispersion analyses, three properties must be distinguished: support of the retarded extension on $t\geq0$, holomorphy of its Fourier--Laplace transform in the upper half-plane $\HH$, and boundary regularity sufficient for a KK integral. The zero-past convention gives the first, but the second and third require separate hypotheses. In this paper ``causal'' means the retarded, linear time-invariant structure of the reduced equation and its stated analytic consequences; it does not mean relativistic signal propagation.

The distinction is particularly important for an NZ reduction. A homogeneous linear, time-invariant equation requires a fixed projection/reference pair and a preparation for which the projected initial component vanishes; factorization removes the inhomogeneous term $\II(t)$ but does not by itself prove any boundary estimate. The spectral input is separate: the physical partial-trace projection is generally non-orthogonal, so $Q\mathcal LQ$ need not be self-adjoint even though the full commutator Liouvillian is self-adjoint in the Hilbert--Schmidt inner product. Thus microscopic unitarity does not, by itself, produce a real-axis spectral measure or the integrability needed for KK. Classical causality--dispersion theory~\cite{kubo1957,toll1956,king2009} also requires growth or boundary assumptions, as emphasized in a recent Laplace-domain treatment~\cite{perinelli2024}; exact NZ constructions~\cite{vacchini2010exact} and Gaussian-bath influence-function methods~\cite{ivander2024} address different parts of the open-system problem. Those contributions do not remove the need to identify the spectrum and the coupling-weighted density of the particular projected generator. The question here is which additional, checkable assumptions make the NZ kernel fit that analytic framework, rather than whether a bath label such as Ohmic is sufficient on its own.

Our main result is a sufficient-condition theorem at this NZ interface. Let the system Liouville space be finite-dimensional and assume
\[
  \Ktilde(z)=\int_{\mathbb R}\frac{w(\lambda)}{z-\lambda}\,d\lambda,
  \qquad w\in L^1(\mathbb R;\BB)\cap L^{p_0}(\mathbb R;\BB),
\]
with no singular real-axis part and $p_0>1$. We prove $\Ktilde\in H^p(\HH;\BB)$ for every $1<p\leq p_0$ (Theorem~\ref{thm:hardy}). The nontrivial point for the operator-valued statement is made explicit: a finite Liouville basis, scalar Cauchy-transform estimates, a common exceptional null set, and finite-dimensional norm equivalence yield the Bochner boundary field. This gives the operator-valued principal-value KK identity and its once-subtracted form (Corollary~\ref{cor:subtracted}), with the subtraction point restricted to the stated Lebesgue-point and principal-value conditions. The $H^1$ endpoint is a separate theorem, not a consequence of the $L^{p_0}$ assumption. We state the result as a projected-spectral criterion; no claim is made that every exact NZ kernel, or every continuum bath, supplies such a density.

The paper then separates three transforms that have different analytic protection. The exact projected kernel is conditional on the spectral hypothesis above. The reduced-state transform $\tilde\sigma(z)$ is holomorphic in $\HH$ for every trace-class initial state under microscopic unitary evolution, with $\|\tilde\sigma(z)\|_1\leq \|\rho_0\|_1/\operatorname{Im}z$ (Proposition~\ref{thm:robustness}); this does not place it in the same Hardy class or make it zero-free. This robustness uses microscopic unitary evolution and trace-class input, but it does not require factorization or the projected-spectrum hypothesis. For correlated preparations the exact equation contains an inhomogeneous term. If that term is discarded in a homogeneous fit, the scalar correction is
\[
  \Ktilde_{\rm eff}(z)=\Ktilde(z)+\tilde{\II}(z)\tilde\sigma(z)^{-1},
\]
and the matrix correction is $\tilde{\II}_{\rm mat}\operatorname{adj}\tilde S/\det\tilde S$. These identities are understood on a common connected open domain $D\subseteq\HH$ where the displayed transforms are holomorphic. A simple baseline zero $\zeta\in D$ gives a first-order quotient pole only when $\tilde{\II}^{(1)}(\zeta)\neq0$; at finite perturbation the actual zero and the actual numerator must be checked. In particular, holomorphy of the state transform is compatible with upper-half-plane zeros, so it cannot be used as a zero-free substitute for a kernel Hardy hypothesis.

This distinction leads to a second, conditional result. For a rational reconstruction with a pole $z_0\in\HH$, the local theorem uses an argument-principle homotopy to count $d$ zeros near $z_0$ under a small-residue bound. The global theorem clears the rational denominators and uses Vieta's formula: the zeros of the matrix-polynomial pencil carry total imaginary part $d\operatorname{Im}z_0>0$, so at least one lies in $\HH$. These algebraic zeros become propagator poles only after the scalar clearing and matrix adjugate no-cancellation tests; a determinant zero introduced solely by clearing a denominator is not counted as a physical pole. If the resulting rational resolvent agrees with the transform of a strongly measurable, uniformly bounded completely positive trace-preserving (CPTP) reduced family on a nonempty open subset, the pole is impossible. This is a statement about an exact representation on that open set, not about an arbitrary Pad\'e approximation; a numerical fit is therefore evidence for a hypothesis to be checked, not a substitute for it. The conditional algebraic obstruction is independent of Hardy placement and is therefore useful when the reconstructed object is not known to arise from an absolutely continuous weight.

The remaining results delimit the scope rather than enlarge the main claim. A positive projected weight gives the anti-Herglotz sign in addition to the Hardy assumptions (Proposition~\ref{thm:passivity}); it is not a consequence of passivity without the positive-weight representation. Bounded finite truncations give a Carleman determinacy criterion, not a stability or continuum convergence theorem (Theorem~\ref{thm:carleman}). The finite Jaynes--Cummings calculation has a discrete spectral measure and therefore tests a shifted, distributional KK identity at fixed truncation. Its integrated relative residual decreases from $0.113\%$ to $0.024\%$ under the recorded grid refinement; this is evidence for fixed numerical consistency, not strict $H^p$ verification. The displayed upper-half-plane zeros are only candidate force-fit locations because numerator non-cancellation has not been established. Finally, the accompanying Lean~4 development~\cite{lean4,mathlib} provides machine-checked verification of the paper-specific deductions from explicitly named classical inputs, bringing the reliability of formal proof to this analytic framework; Appendix~\ref{app:lean} records the precise boundary of what is certified.

The paper is organized as follows. Section~\ref{sec:prelim} fixes the NZ and transform conventions and states the operator-class boundary for the unbounded reduction. Sections~\ref{sec:causality} and~\ref{sec:kk} establish retarded support, Hardy placement, and the KK boundary formulas. Section~\ref{sec:robustness} proves state-transform analyticity and develops the force-fit construction, including the exact $N_{\max}=1$ baseline and the higher-truncation candidate zeros. Section~\ref{sec:consequences} proves the rational pole obstruction, while Section~\ref{sec:further} gives the passivity and finite-truncation consequences. Section~\ref{sec:boundaries} reports the shifted finite-matrix check, and Section~\ref{sec:lean-certification} describes the formalization and its evidence boundary before the Discussion states the remaining model-dependent questions.

\section{Preliminaries}
\label{sec:prelim}

\subsection{Nakajima-Zwanzig projection}

We consider a composite Hilbert space $\mathcal{H} = \mathcal{H}_s \otimes \mathcal{H}_b$ with Hamiltonian $\hat{H} = \hat{H}_s + \hat{H}_b + \hat{H}_{sb}$ and total Liouvillian $\LL\,\cdot \equiv [\hat{H},\,\cdot\,]$.
For bounded $H$, the commutator Liouvillian is a bounded self-adjoint operator on Hilbert--Schmidt space by cyclicity of the trace; this includes every finite truncation below. For unbounded $H$, identify $\mathfrak S_2(\mathcal H)$ with $\mathcal H\otimes\overline{\mathcal H}$ and first define the commutator on its natural initial domain,
\[
  \LL_0=H\otimes\mathbbm 1-\mathbbm 1\otimes\overline H,
  \qquad \mathcal D(\LL_0)=\mathcal D(H\otimes\mathbbm 1)\cap\mathcal D(\mathbbm 1\otimes\overline H).
\]
We denote its standard self-adjoint closure by $\LL=\overline{\LL_0}$; the closure domain is not identified with the initial intersection. Here $\overline H$ acts on the conjugate Hilbert space $\overline{\mathcal H}$. Proposition~\ref{thm:robustness} uses the trace-class evolution $X\mapsto e^{-iHt}Xe^{iHt}$ supplied by Stone's theorem. The $P/Q$ reduction requires more data than that statement alone: the partial trace and reference embedding must be bounded, the coupling blocks must extend to bounded maps, and the joint and compressed evolution families must satisfy the pointwise strong derivative laws for every initial vector and every time, together with the stated continuity at the compressed level. These are explicit all-vector operator-class inputs; in a genuinely unbounded problem they are conditional interfaces, not consequences asserted here from Stone or Hille--Yosida. The choice of operator ideal is part of the data: the state result uses $\mathfrak S_1$, whereas the Hardy result uses the finite-dimensional system Hilbert--Schmidt Liouville space. The Hardy theorem itself assumes the projected Cauchy representation directly. Because the physical partial-trace projection is generally non-orthogonal, no self-adjointness or real spectrum of $Q\LL Q$ is inferred from microscopic unitarity.
We use $\LL=[H,\cdot]$ as a frequency Liouvillian and write the physical generator as $\mathfrak{L}=-i\LL$.
The explicit Laplace convention $\tilde f(z)=\int_0^\infty e^{izt}f(t)\,dt$ fixes the signs in the Cauchy representations below.
To keep the Laplace-domain denominator in the standard dispersive form $z-\LL_s-\Ktilde(z)$, we use the rephased kernel and inhomogeneity defined explicitly below. The factors of $-i$ multiplying them in the physical equation are retained rather than absorbed implicitly.

The projection superoperator $\Proj\rho = \TrOp_b[\rho] \otimes \rho_b^{\mathrm{ref}}$ maps onto the relevant (system) sector; its complement is $\Qproj = \mathbbm{1} - \Proj$, satisfying $\Proj^2 = \Proj$, $\Qproj^2 = \Qproj$, and $\Proj\Qproj = 0$~\cite{nakajima1958,zwanzig1960}. We identify
\begin{equation}
\Proj\rho(t)=\sigma(t)\otimes\rho_b^{\mathrm{ref}},
\qquad \sigma(t)=\TrOp_b[\rho(t)],
\label{eq:projection-identification}
\end{equation}
and suppress this fixed bath factor below.
With this convention, the exact reduced dynamics is represented by the GQME below. For bounded generators this is the standard matrix identity; for unbounded generators it is used under the evolution and domain hypotheses stated in Proposition~\ref{prop:lti}:
\begin{equation}
\frac{d}{dt}\Proj\rho(t) = -i\Proj\LL\Proj\rho(t) -i\int_0^t d\tau\, \KK(\tau)\Proj\rho(t-\tau) -i\II(t),
\label{eq:gqme}
\end{equation}
Here the dispersive memory kernel and the rephased initial-correlation term are
\[
\begin{aligned}
\KK(t)&=-i\Proj\LL\Qproj\,e^{-i\Qproj\LL\Qproj t}\,\Qproj\LL\Proj,\\
\II(t)&=\Proj\LL\Qproj\,e^{-i\Qproj\LL\Qproj t}\,\Qproj\rho(0),
\qquad t\geq0.
\end{aligned}
\]
Indeed, the two physical couplings from $\mathfrak{L}=-i\LL$ produce $(-i)^2=-1$, which is exactly $-i$ times the rephased kernel above; the single initial-correlation coupling produces the last factor $-i$ in Eq.~\eqref{eq:gqme}. Whenever support or a whole-line transform is discussed, $\KK$ denotes the retarded extension obtained by setting this expression to zero for $t<0$.
Choose the projection reference to be the normalized stationary bath state, $\rho_b^{\mathrm{ref}}=\rho_b^{\mathrm{eq}}$ with $\TrOp_b\rho_b^{\mathrm{eq}}=1$. For factorized initial states $\rho(0) = \sigma(0) \otimes \rho_b^{\mathrm{eq}}$, one then has $\Proj\rho(0)=\rho(0)$ and hence $\Qproj\rho(0) = 0$. It follows that $\II(t) = 0$ and the GQME reduces to a homogeneous Volterra equation.
When $H_{sb}$ is bounded (in particular, in every finite truncation), or when the product below is trace-class and its partial trace is a bounded system operator, we write
\[
 \Proj\LL\Proj\,[\sigma\otimes\rho_b^{\mathrm{eq}}]
 = [H_s^{\mathrm{eff}},\sigma]\otimes\rho_b^{\mathrm{eq}},
 \qquad
 H_s^{\mathrm{eff}}=H_s+\TrOp_b\!\bigl[H_{sb}(\mathbbm 1_s\otimes\rho_b^{\mathrm{eq}})\bigr],
\]
and define $\LL_s=[H_s^{\mathrm{eff}},\,\cdot\,]$. Thus the bath mean field is absorbed into the system Hamiltonian. If it vanishes, $H_s^{\mathrm{eff}}=H_s$.
For the general operator-class statement, the bounded relevant drift $\LL_s$ is instead part of the explicit interface; no unbounded partial-trace product is being asserted by this display.

Applying Kramers--Kronig relations to the memory kernel presupposes that the reduced response is a linear, time-invariant causal system and that its transform satisfies the required existence, growth, and boundary-integrability hypotheses. Causal support gives the one-sided structure, but the precise analytic continuation and boundary formula depend on the chosen function or distribution space. For the reduced dynamics of an open quantum system the linear, time-invariant structure is not automatic, and it holds in the regime stated next.

\begin{proposition}[Linear, time-invariant structure of the reduced causal response]
\label{prop:lti}
Work on the joint and reduced trace-class Banach spaces. Assume either a bounded-generator setting (in particular, a finite total Hilbert-space truncation), or the following explicit operator-class inputs: the partial trace and reference-state embedding are bounded and compose to the identity on the reduced space; the joint and compressed evolution families satisfy $U(0)=\mathbbm 1$ and $V(0)=\mathbbm 1$, the joint evolution satisfies its pointwise strong derivative law for every initial vector and every time, the compressed evolution is strongly continuous and satisfies its corresponding pointwise derivative law for every initial vector and every time, and the relevant drift and the two coupling blocks extend to bounded maps and agree with the corresponding restrictions of the Liouvillian. These are the hypotheses encoded by the operator-class data structure; the strong Duhamel formula follows from its printed continuity/differentiability interface. Let the projection reference be a normalized stationary bath state, $\rho_b^{\mathrm{ref}}=\rho_b^{\mathrm{eq}}$, $\TrOp_b\rho_b^{\mathrm{eq}}=1$, and $[\hat H_b,\rho_b^{\mathrm{eq}}]=0$, and let $\rho(0) = \sigma(0)\otimes\rho_b^{\mathrm{eq}}$. Then the exact reduced dynamics~\eqref{eq:gqme} is the homogeneous Volterra equation $\dot{\sigma} = -i\LL_s\sigma-i\KK*\sigma$, which is:
(i)~\emph{linear} in the reduced initial condition $\sigma(0)$, as inherited directly from the linear joint evolution and partial trace. If it is also identified with a Volterra solution in a class where uniqueness holds, that solution map is complex-linear;
(ii)~\emph{time-invariant}, the kernel depending on the time difference $t-\tau$ alone, because the generator and the projection/reference data are time independent and $\KK(t) = -i\Proj\LL\Qproj\,e^{-i\Qproj\LL\Qproj t}\,\Qproj\LL\Proj$; and
(iii)~\emph{causal}, with $\operatorname{supp}\KK\subseteq[0,\infty)$ for the retarded zero extension specified above.
Consequently the reduced response is a linear, time-invariant causal system. Under strong measurability and growth hypotheses sufficient for the one-sided transform to exist, that transform is holomorphic on its domain. Under the separate boundary-integrability hypotheses of Theorems~\ref{thm:hardy} and~\ref{thm:kk}, it has the stated principal-value or $H^1$ Kramers--Kronig boundary values. Singular spectral measures instead require a distributional formulation. None of these conclusions requires a linear-response interpretation of an external drive.
\end{proposition}

\vspace{-18pt}
\begin{proof}
In the bounded setting, the ordinary matrix variation-of-constants identity gives the $\Proj/\Qproj$ split and Eq.~\eqref{eq:gqme}; in the general statement the strong Duhamel formula follows from the printed continuity/differentiability interface. The normalization and matching-reference assumptions imply $\Proj\rho(0)=\rho(0)$, hence $\Qproj\rho(0)=0$ and $\II(t)=0$. Linearity follows from joint evolution and partial trace. The time-independent generator and fixed projection/reference data make $\KK$ a function of elapsed time alone, while causal support is that of the retarded zero extension defined after Eq.~\eqref{eq:gqme}. The transform and boundary conclusions use the additional hypotheses stated in the proposition.
\end{proof}

For correlated preparations $\Qproj\rho(0)$ need not vanish. The exact equation then contains the prescribed inhomogeneity $\II(t)$, so there is in general no autonomous homogeneous reduced map determined by $\sigma(0)$ alone. Force-fitting such a trajectory to a homogeneous equation creates the effective object studied in Section~\ref{sec:eff-kernel}; it does not alter the support convention of the exact NZ kernel.

\subsection{Hardy spaces and the Nevanlinna class}

For $0 < p \leq \infty$, the Hardy space $\Hp{p}(\HH)$ on the upper half-plane $\HH = \{z \in \mathbb{C} : \operatorname{Im}(z) > 0\}$ consists of analytic functions $F : \HH \to \mathbb{C}$ satisfying
\begin{equation}
\|F\|_{\Hp{p}} \equiv \sup_{y>0}\left(\int_{-\infty}^{\infty} |F(x+iy)|^p\, dx\right)^{1/p} < \infty.
\label{eq:hardy}
\end{equation}
For $p=\infty$ the usual modification is understood, with the essential supremum replacing the $L^p$ integral.
Functions in $\Hp{1}$ have non-tangential boundary values a.e.\ on $\mathbb{R}$ and satisfy the Cauchy integral formula on semicircular contours in $\HH$.
The vector-valued Hardy space $\Hp{p}(\BB)$ replaces $|F|$ by the operator norm $\|F\|_\opnorm$ throughout; the Banach space $\BB = \mathcal{B}(\mathcal{B}_{\mathrm{HS}}(\mathcal{H}_s))$ is that of bounded superoperators on the system's Hilbert-Schmidt Liouville space, equipped with the operator norm.
For finite-dimensional $\mathcal{H}_s$, the operator norm and the Hilbert-Schmidt norm are equivalent up to a dimension-dependent constant; all $\Hp{p}$ statements are equivalent between the two choices in that case.
Boundary values and principal-value integrals of $\BB$-valued functions are understood in the Bochner sense~\cite{diestel1977}; by polarization, equivalent statements hold against arbitrary pairs of Hilbert-Schmidt test vectors.

The Herglotz-Nevanlinna class $\mathcal{N}$ consists of functions $f : \HH \to \mathbb{C}$ analytic in $\HH$ with $\operatorname{Im}[f(z)] \geq 0$.
Every such function admits the \emph{Nevanlinna integral representation}~\cite{gesztesy2000}:
\begin{equation}
f(z) = \alpha z + \beta + \int_{-\infty}^{\infty} \left[\frac{1}{\lambda - z} - \frac{\lambda}{1+\lambda^2}\right] d\mu(\lambda),
\label{eq:nevanlinna}
\end{equation}
with $\alpha \geq 0$, $\beta \in \mathbb{R}$, and $\mu$ a positive Borel measure satisfying $\int d\mu(\lambda)/(1+\lambda^2) < \infty$.
Functions with $\operatorname{Im}[f(z)] \leq 0$ (the \emph{anti-Herglotz} class, $-\mathcal{N}$) are obtained by negating Eq.~\eqref{eq:nevanlinna}: the measure changes sign and $\alpha \leq 0$.
Under the positivity and spectral hypotheses of Proposition~\ref{thm:passivity}, every scalar quadratic form $-\langle\!\langle\xi,\Ktilde(z)\xi\rangle\!\rangle$ belongs to $\mathcal N$ and admits the scalar representation~\eqref{eq:nevanlinna}. Equivalently, $-\Ktilde$ is an operator-valued Herglotz function. This terminology does not identify the operator-valued function itself with the scalar class just defined.

\section{Retarded NZ Kernel and Hardy Placement}
\label{sec:causality}

\begin{proposition}[Retarded support and spectral decay]
\label{thm:causality}
Let $\KK(t)$ be the projected kernel defined by Eq.~\eqref{eq:gqme} for $t\geq0$, and use its retarded zero extension on the full real line. Then $\operatorname{supp}\KK\subseteq[0,\infty)$. If, in addition, the projected kernel admits the real-axis coupling-weighted spectral representation
\begin{equation}
\KK(t)=-i\int_{\mathbb{R}} e^{-i\lambda t} w(\lambda)\,d\lambda,
\qquad w\in L^1(\mathbb{R};\BB),
\label{eq:thm-causality-rep}
\end{equation}
then the kernel decays, $\|\KK(t)\|_\opnorm \to 0$ as $t \to \infty$, and its Laplace transform is the Cauchy transform of $w$, that is, Eq.~\eqref{eq:spectral-rep} below holds with the same $w$.
\end{proposition}

\begin{proof}
The support statement follows from the retarded extension: the projected formula defines $\KK(t)$ for nonnegative elapsed time, and the extension is zero for $t<0$. Factorization is needed in Proposition~\ref{prop:lti} to remove $I(t)$ and obtain a homogeneous reduced equation; it is not needed to define this support.

Decay is not a consequence of bounded projected propagation alone; it follows from the assumed absolutely continuous spectral representation.
By the Riemann--Lebesgue lemma, $\|\KK(t)\|_\opnorm = \|\int e^{-i\lambda t} w(\lambda)\, d\lambda\|_\opnorm \to 0$ as $t \to \infty$ (the unimodular prefactor $-i$ does not affect the norm).

For the Cauchy representation, substitute Eq.~\eqref{eq:thm-causality-rep} into $\Ktilde(z)=\int_0^\infty e^{izt}\KK(t)\,dt$ and use, for $\operatorname{Im}z>0$,
\[
\int_0^\infty e^{i(z-\lambda)t}\,dt=\frac{i}{z-\lambda},
\qquad\text{whence}\qquad
\Ktilde(z)=\int_{\mathbb{R}}\frac{w(\lambda)}{z-\lambda}\,d\lambda .
\]
The $-i$ in Eq.~\eqref{eq:thm-causality-rep} and the $+i$ from the Laplace integral cancel, so the Cauchy representation follows with no redefinition of $w$.
\end{proof}

\begin{remark}[Decay rate depends on analytic structure, not a spectral gap]
A compact-support or gap condition on $w$ alone does not determine a decay rate. Exponential decay can be obtained from an analytic continuation into a strip, together with the bounds needed to deform the Fourier contour. Such strip control is an additional regularity assumption, not a consequence of a real-axis gap.
\end{remark}

\begin{theorem}[Hardy Placement under a real-axis spectral hypothesis]
\label{thm:hardy}
Let the system Liouville space be finite-dimensional.
Assume that, for the chosen projection/inner-product pair, the projected kernel has the Cauchy-transform representation
\begin{equation}
\Ktilde(z) = \int_{-\infty}^{\infty} \frac{w(\lambda)}{z - \lambda}\,d\lambda,
\qquad z\in\HH,
\label{eq:spectral-rep}
\end{equation}
with no singular real-axis part.
Assume further that there exists $p_0 > 1$ such that $w \in L^1(\mathbb{R};\BB)\cap L^{p_0}(\mathbb{R};\BB)$.
Then $\Ktilde$ belongs to $\Hp{p}(\BB)$ for every finite $p$ with $1 < p \leq p_0$.
\end{theorem}

\begin{proof}
Since $w\in L^1\cap L^{p_0}$, interpolation gives $w\in L^p$ for every $1<p\leq p_0$.
Fix a basis $\{B_a\}_{a=1}^{D}$ of the finite-dimensional space $\BB$, with coordinate functionals $\ell_a$. Write $w(\lambda)=\sum_a w_a(\lambda)B_a$, where $w_a=\ell_a\circ w\in L^p(\mathbb R)$. For each $a$, the scalar Cauchy-transform theorem gives
\[
 F_a(z)=\int_{\mathbb R}\frac{w_a(\lambda)}{z-\lambda}\,d\lambda\in H^p(\HH),
 \qquad \|F_a\|_{H^p}\leq C_p\|w_a\|_{L^p}.
\]
The boundary formula for each component holds off a null set $N_a$. Since the basis is finite, $N=\bigcup_aN_a$ is still null, so all componentwise boundary identities hold simultaneously off $N$.

Finite-dimensional norm equivalence gives constants $c_D,C_D>0$ such that
\[
 c_D\max_a|x_a|\leq\Bigl\|\sum_ax_aB_a\Bigr\|_\opnorm
 \leq C_D\sum_a|x_a|.
\]
Applying the upper bound to $\Ktilde(z)=\sum_aF_a(z)B_a$, followed by the finite-sum inequality in $L^p$, yields
\[
 \sup_{y>0}\|\Ktilde(\,\cdot+iy)\|_{L^p(\mathbb R;\BB)}
 \leq C(D,p)\|w\|_{L^p(\mathbb R;\BB)}.
\]
Thus $\Ktilde\in H^p(\HH;\BB)$. This componentwise proof is sufficient because $D<\infty$; no infinite-dimensional vector-valued singular-integral theorem is required.

The exclusion of $p = 1$ is an endpoint limitation of this general theorem: the Cauchy transform is bounded from $L^p(\mathbb{R})$ to $H^p(\HH)$ for $1 < p < \infty$ but the Hilbert transform is not bounded on all of $L^1$.
Thus a standard $H^1$ KK theorem may still hold for a special model, but it requires a separate endpoint estimate rather than following from the present $L^{p_0}$ hypothesis.
For threshold singularities or algebraic tails, the safer general statement is the $L^p$/principal-value or suitably subtracted form in Corollary~\ref{cor:subtracted}.
\end{proof}

\begin{remark}[Verifying integrability in specific models]
The theorem is stated only in terms of the projected weight $w$. In controlled weak-coupling constructions for linearly coupled Gaussian baths, the leading weight can sometimes be read from the bath spectral density. For example, a threshold $w(\lambda)\sim|\lambda|^{s-1}$ with $0<s<1$ permits only $p<1/(1-s)$ near the origin. Beyond such a controlled setting, identifying the operator-valued $L^p$ weight for labels such as Ohmic or Drude--Lorentz is a concrete model-dependent task.
\end{remark}

\begin{remark}[On the spectral hypothesis and its verification]
\label{rem:spectral-hypothesis}
Theorem~\ref{thm:hardy} is conditional: it assumes the projected generator admits a real-axis spectral representation with $w\in L^{p_0}$.
For the standard Nakajima--Zwanzig partial-trace projection $P\rho = \TrOp_b[\rho] \otimes \rho_b^{\mathrm{ref}}$, this is an additional assumption, not a consequence of microscopic unitarity alone.

Numerical inspection of the projected generator $\Qproj\LL\Qproj$ in truncated models confirms that the real-spectrum condition is model-dependent; the finite-grid counts and thresholds are recorded in Appendix~\ref{app:qlq-spectrum}. These tests establish neither a general Jaynes--Cummings theorem nor a continuum phase boundary. They show only that real spectrum of the physical projected generator is not automatic. In non-perturbative settings, verifying the $L^p$ condition requires identifying the exact operator-valued projected weight, which remains a model-dependent task.
Theorem~\ref{thm:hardy} therefore provides an explicit sufficient-condition framework within which KK relations are rigorously valid. The numerical checks retained here are the projected-generator scans in Appendix~\ref{app:qlq-spectrum} and a shifted-line operator-valued consistency test of an extracted $4\times4$ Jaynes--Cummings kernel. They test the real-spectrum condition on finite grids and the distributional boundary structure at fixed truncation; they do not establish strict $H^p$ membership for a finite discrete bath.
\end{remark}

\begin{remark}[Finite real-spectrum truncations and recurrence]
\label{rem:finite}
Finite dimensionality alone does not imply the real-axis spectral representation used in Theorem~\ref{thm:hardy}: a non-orthogonal NZ projection can produce a non-normal $\Qproj\LL\Qproj$ with complex eigenvalues or Jordan blocks. A finite truncation that additionally admits a diagonalizable real-axis atomic representation has
\[
\KK(t)=-i\sum_n w_ne^{-i\lambda_nt},
\qquad
\Ktilde(z)=\sum_n\frac{w_n}{z-\lambda_n},
\qquad \lambda_n\in\mathbb R.
\]
After equal frequencies are combined, any surviving nonzero atom excludes membership in $\Hp{p}$ for every $p>1$, the range of Theorem~\ref{thm:hardy}. Indeed, a nonzero simple boundary pole gives
\[
\sup_{y>0}\bigl\|(\,\cdot+iy-\lambda_n)^{-1}\bigr\|_{L^p(\mathbb R)}=\infty,
\]
and no bounded contribution from the other frequencies removes this local divergence. At $p=1$ the same local divergence is logarithmic. For $0<p<1$, membership can depend on cancellations in the high-frequency moments; that range is not used here.
Its KK relation is therefore distributional, and a shifted-line reconstruction is well defined away from the singular frequencies. A general finite truncation yields only a rational or meromorphic transform; complex projected eigenvalues and Jordan blocks can instead produce off-axis or higher-order poles, so no real-axis distributional KK statement follows without the added spectral assumptions.
\end{remark}

\section{Operator-Valued Dispersion Relations}
\label{sec:kk}

\subsection{Kramers--Kronig relations}

Fix once and for all a Hilbert--Schmidt matrix-unit basis $\mathcal E$ for the system Liouville space. If $F=[F_{ab}]_{\mathcal E}$ is operator-valued, write $\operatorname{Re}_{\mathcal E}F=[\operatorname{Re}F_{ab}]_{\mathcal E}$ and $\operatorname{Im}_{\mathcal E}F=[\operatorname{Im}F_{ab}]_{\mathcal E}$. These are entrywise real and imaginary parts in the fixed basis, not the Hermitian and anti-Hermitian parts of the superoperator.

\begin{theorem}[Operator-valued KK boundary identity]
\label{thm:kk}
Let $\KK(t)$ be causal and absolutely integrable with $\Ktilde \in \Hp{1}(\BB)$.
Then the standard KK relations hold a.e.:
\begin{equation}
\operatorname{Re}_{\mathcal E}\Ktilde(\omega)
= \frac{1}{\pi}\PP\!\int_{-\infty}^{\infty}\!\! d\omega'\,
\frac{\operatorname{Im}_{\mathcal E}\Ktilde(\omega')}{\omega' - \omega}.
\label{eq:kk}
\end{equation}
\end{theorem}

\begin{corollary}[Principal-value and subtracted KK beyond $H^1$]
\label{cor:subtracted}
If $\Ktilde \in \Hp{p}(\BB)$ for some finite $p>1$, Eq.~\eqref{eq:kk} holds in the $L^p$ boundary-value/principal-value sense. Let $\omega_*$ be a common Lebesgue point of the boundary components at which the principal-value identity holds and $\Ktilde(\omega_*)$ is finite. Then, for almost every $\omega$ in the same full-measure set, subtraction of the two principal-value identities gives
\begin{equation}
\operatorname{Re}_{\mathcal E}\Ktilde(\omega)
-\operatorname{Re}_{\mathcal E}\Ktilde(\omega_*) =
\frac{\omega-\omega_*}{\pi}\PP\!\int_{-\infty}^{\infty}\!\! d\omega'\,
\frac{\operatorname{Im}_{\mathcal E}\Ktilde(\omega')}
{(\omega'-\omega_*)(\omega' - \omega)} .
\label{eq:kk-subtracted}
\end{equation}
If the right-hand side is used as an absolutely convergent integral rather than as the difference of two principal values, that convergence is an additional assumption. In particular, $\omega_*$ cannot be placed at a threshold where the boundary value diverges.
\end{corollary}

Theorem~\ref{thm:hardy} delivers $\Hp{p}(\BB)$ membership only for $p>1$, so on the models treated here the operative dispersion relation is the $L^p$/principal-value form of Corollary~\ref{cor:subtracted}, in its once-subtracted variant when the tail is algebraic.
Theorem~\ref{thm:kk} is stated at the classical $H^1$ endpoint, where the relation takes its standard almost-everywhere pointwise form; the corollary is what covers the whole admissible range $1<p\leq p_0$ produced by Theorem~\ref{thm:hardy}.

\begin{proof}[Proof of Theorem~\ref{thm:kk}]
We separate the scalar input from its finite-dimensional componentwise lift.

\emph{Scalar input (cited).}
For a scalar $f\in H^1(\HH)$, the a.e.\ non-tangential boundary value and the scalar Hardy projection theorem give
\begin{equation}
\operatorname{Re} f(\omega) \;=\; \frac{1}{\pi}\,\operatorname{p.v.}\!\!\int_{-\infty}^{\infty}\!\! d\omega'\,\frac{\operatorname{Im} f(\omega')}{\omega'-\omega},
\qquad \text{a.e. }\omega,
\label{eq:scalar-kk}
\end{equation}
equivalently $\operatorname{Im} f = -\mathcal{H}(\operatorname{Re} f)$ with the present sign convention; this is the classical scalar Kramers--Kronig relation~\cite{king2009,toll1956,kubo1957,perinelli2024}.

\emph{Finite-dimensional operator-valued lift.}
Let $\Ktilde\in\Hp{1}(\BB)$. In the fixed basis $\mathcal E$, every coordinate $f_{ab}(z)=[\Ktilde(z)]_{ab}$ is a scalar $H^1$ function and therefore has a non-tangential $L^1$ boundary value and obeys Eq.~\eqref{eq:scalar-kk} outside a null set $N_{ab}$. There are finitely many coordinates, so $N=\bigcup_{ab}N_{ab}$ is null and all scalar identities hold simultaneously on $\mathbb R\setminus N$. Finite-dimensional norm equivalence reassembles the coordinate limits into the Bochner $L^1$ boundary value and gives
\[
\operatorname{Re}_{\mathcal E}\Ktilde(\omega) \;=\; \frac{1}{\pi}\,\operatorname{p.v.}\!\!\int_{-\infty}^{\infty}\!\! d\omega'\,\frac{\operatorname{Im}_{\mathcal E}\Ktilde(\omega')}{\omega'-\omega},
\qquad \text{a.e. }\omega,
\]
where the principal value is the finite-dimensional Bochner principal value obtained by assembling the scalar limits. No additional matrix-valued singular-integral theorem is needed.
\end{proof}

\begin{proof}[Proof of Corollary~\ref{cor:subtracted}]
For $\Ktilde\in\Hp{p}(\BB)$ with $1<p<\infty$, the scalar Hardy projection theorem gives $L^p$ convergence of the truncated Hilbert transforms in every coordinate. The finite-basis argument above then gives Eq.~\eqref{eq:kk} in $L^p$ and almost everywhere after choosing a common representative. At a subtraction point $\omega_*$ satisfying the hypotheses of the corollary, apply this identity at $\omega$ and $\omega_*$ and subtract. The common boundary field is integrated against
\[
\frac{1}{\omega'-\omega}-\frac{1}{\omega'-\omega_\ast}
\;=\;\frac{\omega-\omega_\ast}{(\omega'-\omega)(\omega'-\omega_\ast)},
\]
which gives Eq.~\eqref{eq:kk-subtracted}. This derivation treats the right-hand side as the difference of two principal values. Absolute convergence, when claimed, must be checked separately.
\end{proof}

\section{State Transforms and Force-Fit Kernels}
\label{sec:robustness}

\subsection{Analyticity of the reduced-state transform}

The memory kernel is fixed by the Hamiltonian, the chosen projection, and the reference state. The inner product enters the spectral realization used to analyze that kernel, but is not an additional dependence of the exact NZ construction. Factorization of the initial state is used only to remove $\II(t)$ and obtain the homogeneous LTI equation. Realistic preparations --- thermal equilibrium of the joint system, post-measurement conditional states, and post-quench initial states --- are generally \emph{correlated}, so it is important to separate the analyticity of their state trajectories from Hardy placement of the kernel.

Here we show that the reduced-state Laplace transform is holomorphic on $\HH$ regardless of initial correlations. This is a statement about the state trajectory, not an $H^p$ placement theorem and not a zero-free theorem.

\begin{proposition}[State analyticity for arbitrary trace-class initial data]
\label{thm:robustness}
Let $H$ be a self-adjoint total Hamiltonian on $\mathcal{H}_s \otimes \mathcal{H}_b$, let $X_0$ be any trace-class operator on this joint space, and set $\sigma_{X_0}(t)=\TrOp_b[e^{-iHt}X_0e^{iHt}]$. Then $t\mapsto\sigma_{X_0}(t)$ is continuous in trace norm, $\|\sigma_{X_0}(t)\|_1\leq\|X_0\|_1$, and
\[
\tilde{\sigma}_{X_0}(z) = \int_0^\infty e^{izt}\,\sigma_{X_0}(t)\,dt
\]
is analytic in $\HH$ as a trace-class-valued function, with
\[
 \|\tilde{\sigma}_{X_0}(z)\|_\opnorm\leq\|\tilde{\sigma}_{X_0}(z)\|_1
 \leq\frac{\|X_0\|_1}{\operatorname{Im}z}.
\]
In particular, for any factorized or correlated density operator $X_0=\rho_0$, one has $\|\sigma_{\rho_0}(t)\|_1=1$ and the normalized bound $1/\operatorname{Im}z$.
\end{proposition}

\begin{proof}
Unitary conjugation preserves the trace norm, and the partial trace is trace-norm contractive, so $\|\sigma_{X_0}(t)\|_1\leq\|X_0\|_1$. For a density operator, positivity and trace preservation give equality with $1$. For unbounded $H$, Stone's theorem gives strong continuity of $U(t)=e^{-iHt}$; approximation of $X_0$ by finite-rank operators then gives trace-norm continuity of $U(t)X_0U(t)^\dagger$, and the partial trace is trace-norm continuous. For $z=x+iy$, $y>0$, the trace-class Bochner integral therefore exists and obeys
\[
 \|\tilde\sigma_{X_0}(z)\|_1\leq\int_0^\infty e^{-yt}\|\sigma_{X_0}(t)\|_1\,dt
 \leq\frac{\|X_0\|_1}{y}.
\]
Local uniform domination permits differentiation under the integral, proving trace-class-valued holomorphy; the operator-norm bound follows from $\|A\|_\opnorm\leq\|A\|_1$.
\end{proof}

\begin{remark}[Kernel vs.\ reduced-state Laplace transform]
\label{rem:kernel-vs-state}
Proposition~\ref{thm:robustness} concerns $\tilde{\sigma}(z)$, not the memory kernel.
Strictly, $\tilde{\sigma}(z)$ is the Laplace transform of a \emph{state trajectory}: it is a propagator map on the reduced initial state only when the initial state is factorized.
For correlated preparations the map from $\sigma(0)$ to $\sigma(t)$ is not well-defined without an assignment rule (and may be nonlinear or nonunique under such a rule), although $\sigma(t)$ remains linear in the full initial state $\rho_0$.
The NZ memory kernel $\KK(t)$, in contrast, depends on the Hamiltonian and on the chosen projection/reference data, but not on the particular initial state once those data are fixed. Its Hardy membership is a property of the projected construction together with the spectral realization used to represent it, not of $H$ alone.
The initial state enters the exact reduced equation through the rephased inhomogeneous term $\II(t) = \Proj\LL\Qproj\,e^{-i\Qproj\LL\Qproj t}\,\Qproj\rho_0$, appearing as $-i\II(t)$ in Eq.~\eqref{eq:gqme}. Holomorphy does not forbid $\tilde\sigma(z)$ from having zeros in $\HH$.
\end{remark}

\subsection{Force-fit kernels with initial correlations}
\label{sec:eff-kernel}

Proposition~\ref{thm:robustness} shows that the reduced-state Laplace transform is holomorphic in $\HH$, although it may have zeros there.
A separate question is what happens to the \emph{memory kernel} when an experimentalist (or numerical extractor) force-fits the correlated dynamics into a homogeneous master equation, discarding the inhomogeneous term $\II(t)$.
This is the object produced by that homogeneous fitting procedure; it is the object whose dispersion diagnostic can differ from that of the exact kernel, and it is not protected by the state-boundedness theorem.

\subsubsection{Force-fit definition}

Starting from the exact inhomogeneous GQME
\begin{equation}
\frac{d}{dt}\sigma(t) = -i\LL_s\sigma(t) -i(\KK * \sigma)(t) -i\II(t)
\label{eq:exact-gqme}
\end{equation}
fix a connected open domain $D\subseteq\HH$ on which the transforms $\Ktilde$, $\tilde{\II}$, and $\tilde\sigma$ (or the stacked matrix $\tilde S$ below) exist and are holomorphic. All force-fit identities in this subsection are statements on $D$; existence of $\tilde{\II}$ on the whole upper half-plane is not inferred from state boundedness or from Theorem~\ref{thm:hardy}. Demanding a purely homogeneous form $\dot{\sigma} = -i\LL_s\sigma-i\KK_\mathrm{eff}*\sigma$, Laplace transformation (with the convention $\tilde f(z)=\int_0^\infty e^{izt}f(t)\,dt$, so that $\widetilde{\dot f}(z)=-iz\tilde f(z)-f(0)$) gives, in a scalar channel,
\begin{equation}
\boxed{\quad\Ktilde_\mathrm{eff}(z) = \Ktilde(z) + \tilde{\II}(z)\,\tilde{\sigma}(z)^{-1},\quad}
\label{eq:Keff}
\end{equation}
where the inverse means division by the scalar channel transform.
For a full operator-valued reconstruction, let $d=\dim\mathcal B_{\mathrm{HS}}(\mathcal H_s)$ and choose $d$ linearly independent trajectories. With vectorization in a fixed Liouville basis, set
\[
  \tilde S(z)=\bigl[\tilde\sigma^{[1]}(z)\ \cdots\ \tilde\sigma^{[d]}(z)\bigr]\in\mathbb C^{d\times d},\qquad
  \tilde{\II}_{\rm mat}(z)=\bigl[\tilde{\II}^{[1]}(z)\ \cdots\ \tilde{\II}^{[d]}(z)\bigr]\in\mathbb C^{d\times d},
\]
where each column is the vectorized state or inhomogeneity transform for the corresponding preparation. The matrix force-fit identity is then
\[
  \Ktilde_\mathrm{eff}(z)=\Ktilde(z)+\tilde{\II}_{\rm mat}(z)\tilde S(z)^{-1}.
\]
Assume that $\det\tilde S$ is not identically zero, equivalently that $\tilde S$ is invertible at least at one point of $D$. The inverse is holomorphic on $\{z\in D:\det\tilde S(z)\neq0\}$ and extends meromorphically through isolated determinant zeros by $\tilde S^{-1}=\operatorname{adj}\tilde S/\det\tilde S$. At such a zero $\zeta$, write $\nu_\zeta(f)$ for the vanishing order of a holomorphic scalar entry and set $\nu_\zeta(0)=+\infty$. The $(a,b)$ entry of the correction has pole order
\[
  \max\!\left\{0,\nu_\zeta(\det\tilde S)-
  \nu_\zeta\!\left((\tilde{\II}_{\rm mat}\operatorname{adj}\tilde S)_{ab}\right)\right\}.
\]
Thus a determinant zero produces a pole only when at least one relevant adjugate numerator fails to cancel it.

\subsubsection{Perturbative expansion}

Fix the Hamiltonian, projection, reference state, and hence the exact kernel. In such a scalar channel, let $\epsilon\in\mathbb R$ parametrize a family of initial states (rather than a change of the Hamiltonian):
\[
\rho_\mathrm{fact}=\sigma_s^{(0)}\otimes\rho_b^{\mathrm{ref}},
\qquad
\rho_0(\epsilon) = \rho_\mathrm{fact} + \epsilon\,\delta\rho + \epsilon^2\,\rho^{(2)} + O(\epsilon^3).
\]
The coefficient operators are assumed trace class and Hermitian, and the family is assumed to remain positive with unit trace for sufficiently small $|\epsilon|$. The expansion is understood in a topology that gives the corresponding local-uniform transform expansion. Since $\Qproj\rho_\mathrm{fact}=0$ for this reference-factorized state,
\[
\tilde{\II}(z;\epsilon) = \epsilon\,\tilde{\II}^{(1)}(z) + \epsilon^2\,\tilde{\II}^{(2)}(z) + O(\epsilon^3),
\]
where the coefficient transforms are assumed holomorphic on $D$. This is a separate existence hypothesis on the inhomogeneity; the kernel-weight assumption in Theorem~\ref{thm:hardy} does not by itself imply it. The Taylor expansion of Eq.~\eqref{eq:Keff} gives
\begin{equation}
\Ktilde_\mathrm{eff}(z;\epsilon) = \Ktilde(z) + \epsilon\,\tilde{\II}^{(1)}(z)\,\tilde{\sigma}_0(z)^{-1} + O(\epsilon^2),
\label{eq:Keff-pert}
\end{equation}
where $\tilde{\sigma}_0(z)$ is the transform generated by $\rho_\mathrm{fact}$. More explicitly, assume $\tilde\sigma(z;\epsilon)=\tilde\sigma_0(z)+O(\epsilon)$ locally uniformly on compact subsets of $D$ (and in the boundary norm whenever the boundary version of this expansion is used). On compact sets where $\tilde\sigma_0$ is bounded away from zero, the omitted inverse correction starts at $O(\epsilon^2)$ because the inhomogeneous numerator itself is $O(\epsilon)$.

\subsubsection{The zero-of-\texorpdfstring{$\tilde{\sigma}$}{sigma} mechanism}

Both $\tilde{\II}^{(1)}(z)$ and $\tilde{\sigma}_0(z)$ are analytic in $D$. Their quotient is analytic away from the zeros of $\tilde{\sigma}_0$. At a simple zero $\zeta\in D$, the first-order correction in Eq.~\eqref{eq:Keff-pert} has a pole when $\tilde{\II}^{(1)}(\zeta)\neq0$. For a nonzero finite perturbation the zero generally moves, so an exact force-fit pole requires checking the actual zero of $\tilde\sigma(z;\epsilon)$ and the actual numerator there (or proving persistence by a separate local argument).
The CPTP obstruction of Theorem~\ref{thm:cp-hardy} applies only after the relevant rationality, local or global pencil, and propagator no-cancellation hypotheses of that theorem have also been checked.
For a matrix reconstruction the identity $\tilde S^{-1}=\operatorname{adj}(\tilde S)/\det\tilde S$ gives the same meromorphic conclusion. What remains model-dependent is the location and multiplicity of determinant zeros and whether $\tilde{\II}_{\rm mat}\,\operatorname{adj}(\tilde S)$ cancels them.

For systems with near-resonant structure, channels with no positivity constraint --- in particular, coherence components $\tilde{\sigma}_0^{(eg)}(z)$ --- can admit zeros in the open upper half-plane.
Proposition~\ref{prop:sigma-zeros} gives the exact $N_{\max}=1$ baseline after removable factors are cancelled, under the preparation hypotheses printed in the proposition. For nonzero coupling, the coherence channel has the single real zero $z=\omega_c$ when the bath state is Fock-diagonal, while the population zeros depend explicitly on the bath occupation when the vacuum sector has nonzero weight. A Fock-coherent bath is outside this baseline, and no $N_{\max}=1$ zero-placement claim is made for it. The decoupled case and the pure $|1\rangle$ bath have no finite zero in the corresponding population channel. For higher truncations, real numerator coefficients imply only conjugate pairing, not real-rootedness. The tested Fock-diagonal cases have real zeros within the working precision, whereas a displayed Fock-coherent $N_{\max}=2$ preparation has an open-upper-half-plane pair. The displayed zero was checked from an exact symbolic numerator and by direct evaluation of the independently assembled raw Liouville resolvent; the full coupling grid remains numerical evidence and is not a consequence of Proposition~\ref{prop:sigma-zeros}.
When such a zero $\zeta\in\HH$ is present, it is only a \emph{candidate} force-fit pole location. It contributes a first-order pole through $\tilde{\II}^{(1)}\tilde{\sigma}_0^{-1}$ only if $\zeta\in D$ and $\tilde{\II}^{(1)}(\zeta)\neq0$; an exact finite-perturbation pole requires the additional check just described. The factorized calculation itself has $\tilde{\II}=0$ and therefore does not exhibit an effective-kernel pole.

\begin{figure}[!tbp]
\centering
\includegraphics[width=\textwidth]{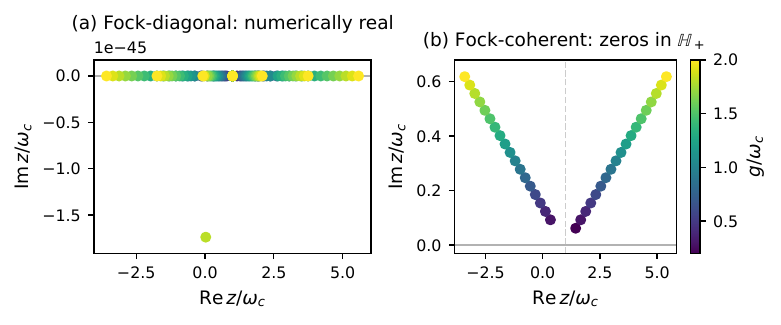}
\caption{Zeros of the coherence-channel transform $\tilde{\sigma}_0^{(eg)}(z)$ for a finite Jaynes--Cummings model ($N_{\max}=2$, $\omega_0=\omega_c=1$), with initial system state $|+\rangle\langle+|$ and a 20-point coupling grid $g/\omega_c\in[0.1,2.0]$.
 (a) \emph{Fock-diagonal bath}: the truncated thermal state with $\beta\omega_c=1$ gives zeros on the real axis within the working precision throughout this grid. This is numerical evidence; real numerator coefficients alone imply only conjugate pairing.
 (b) \emph{Fock-coherent bath} $\rho_b=\tfrac12(|0\rangle+e^{i\chi}|1\rangle)(\langle0|+e^{-i\chi}\langle1|)$, $\chi=\pi/2$ and $g=0.7$: the exact $N_{\max}=2$ numerator and high-precision roots give $\zeta\approx2.54284273+0.21645133i$ and the second upper-half-plane root $-0.54284273+0.21645133i$. Direct evaluation of the independently assembled raw $36\times36$ Liouville resolvent gives a residual of $6.52\times10^{-81}$ at 80-digit precision. These are state-transform zeros and only candidate force-fit pole locations.}
\label{fig:jc-zeros}
\end{figure}

\subsubsection{Exact baseline and numerical zero placement in the Jaynes--Cummings model}
\label{sec:zero-scan}

The $N_{\max}=1$ closed forms, their non-pole zero equivalences, and the conjugation symmetry are algebraic and have been checked in Lean. The removable-factor analysis in the degenerate cases is written explicitly below. Higher-truncation root locations are numerical. For the displayed Fock-coherent case, the root is stable from 30 to 80 digits and the independent raw-resolvent residual decreases from $7.22\times10^{-31}$ to $6.52\times10^{-81}$.

\begin{proposition}[Exact $N_{\max}=1$ zero placement]
\label{prop:sigma-zeros}
Let $\rho(0)=\sigma_s(0)\otimes\rho_b$ be factorized, with real Hamiltonian parameters $\omega_0,\omega_c,g$.
(i)~\emph{Population channel.} Let $N_{\max}=1$, let $\sigma_s(0)=|e\rangle\langle e|$, and write the normalized Fock-diagonal bath as $\rho_b=p_0|0\rangle\langle0|+p_1|1\rangle\langle1|$, where $p_0,p_1\geq0$ and $p_0+p_1=1$. On resonance, if $g\neq0$ and $p_0>0$, the meromorphic continuation of the population transform has precisely the two finite zeros determined by
\[
z^2=2g^2(1+p_1).
\]
Thus $z=\pm g\sqrt{2}$ holds only for the vacuum bath $p_1=0$. If $g=0$ or $p_0=0$, the continuation reduces to a nonzero multiple of $1/z$ and has no finite zero.
(ii)~\emph{Coherence channel.} At $N_{\max}=1$, assume in addition that the bath state is Fock-diagonal, and write $c=\langle e|\sigma_s(0)|g\rangle$. If $c$ and $g$ are both nonzero, the meromorphic continuation has the unique finite zero $z=\omega_c$, for arbitrary detuning. If $g=0$, the factor $z-\omega_c$ cancels and there is no finite zero. No conclusion is asserted here for a Fock-coherent bath.
(iii)~\emph{Conjugation symmetry.} Whenever a coherence-channel numerator polynomial has real coefficients, its zero set is closed under complex conjugation. This statement does not imply that all zeros are real.
\end{proposition}

\begin{proof}
(i)~At $N_{\max}=1$, on resonance and for $\sigma_s(0)=|e\rangle\langle e|$, the population-channel zero condition away from the displayed denominators reads
\[
\frac{p_0}{2}\Bigl(\frac{1}{z}+\frac{z}{z^2-4g^2}\Bigr)+\frac{p_1}{z}=0 .
\]
Away from the displayed denominators we may multiply by $2z(z^2-4g^2)\neq0$, which gives
$p_0\bigl[(z^2-4g^2)+z^2\bigr]+2p_1(z^2-4g^2)=0$, that is, $2z^2(p_0+p_1)=4g^2(p_0+2p_1)$.
Normalization $p_0+p_1=1$ turns this into $z^2=2g^2(1+p_1)$. If $g\neq0$ and $p_0>0$, these roots are nonzero and
\[
z^2-4g^2=-2g^2p_0\neq0,
\]
so neither root was introduced by clearing a denominator. Since $2g^2(1+p_1)=2g^2$ holds exactly when $p_1=0$, the vacuum values $z=\pm g\sqrt{2}$ characterize the vacuum bath and no other occupation. If $g=0$, the two fractions in parentheses combine to $2/z$ after cancellation; if $p_0=0$, only $p_1/z=1/z$ remains. In either degenerate case the transform has no finite zero.
(ii)~With a Fock-diagonal bath, the two excitation blocks contributing to the reduced coherence combine to give
\[
\tilde\sigma_{eg}(z)=\frac{i\,c\,(z-\omega_c)}{(z-\omega_0)(z-\omega_c)-g^2},
\]
where $c=\sigma_{eg}(0)$. This is exactly the Fock-diagonal input encoded by the finite-matrix resolvent calculation. For $c\neq0$ and $g\neq0$, the denominator at $z=\omega_c$ equals $-g^2\neq0$, so the continuation vanishes there and nowhere else. If $g=0$, the factor $z-\omega_c$ cancels before zeros are counted, leaving $ic/(z-\omega_0)$ and hence no finite zero.
(iii)~A polynomial with real coefficients satisfies $p(\bar z)=\overline{p(z)}$, so its zero set is invariant under conjugation; this constrains the zeros to be real or to come in conjugate pairs, and forces neither.
\end{proof}

The banded-overlap count and the high-frequency sum rule below are exact algebraic constraints behind this placement; higher-truncation root locations are evaluated separately with 80-digit arithmetic because standard floating-point root-finding is ill-conditioned for the high-degree coherence numerator.
\emph{Banded overlap structure.} Write the eigendecomposition $\hat H=\sum_\mu E_\mu|\varphi_\mu\rangle\langle\varphi_\mu|$ and define the coherence-channel overlap $O_{\mu\nu}=\sum_n\langle e,n|\varphi_\mu\rangle\langle\varphi_\nu|g,n\rangle$. Because the rotating-wave coupling conserves total excitation number, $O_{\mu\nu}$ is nonzero only when $\varphi_\mu$ lies in excitation manifold $m{+}1$ and $\varphi_\nu$ in manifold $m$. For $N_{\max}=6$ this selects exactly $24$ of the $14^2=196$ energy-eigenstate pairs, and the transform has the partial-fraction representation
\begin{equation}
\tilde{\sigma}_{eg}(z)=\sum_{\mu\nu}a_{\mu\nu}\,\frac{1}{z-(E_\mu-E_\nu)},
\label{eq:sigma-partial-fraction}
\end{equation}
where $[\rho_0]_{\mu\nu}$ denotes the matrix element of the initial state in this energy eigenbasis and $a_{\mu\nu}=i[\rho_0]_{\mu\nu}O_{\mu\nu}$ before equal Bohr frequencies are combined. It has at most $24$ distinct simple poles, rather than all $14^2$ pairs. After terms at equal Bohr frequencies are combined and zero combined residues are removed, let $r\leq24$ be the number of surviving distinct poles. The reduced denominator then has degree $r$; the special value $r=24$ additionally requires distinct selected Bohr frequencies and a nonzero combined residue at each one.
\emph{High-frequency sum rule.} With the present $e^{izt}$ convention, integration by parts gives $\lim_{z\to\infty}z\,\tilde{\sigma}_{eg}(z)=i\sigma_{eg}(0)$ along non-tangential rays in $\HH$. Hence the residues obey
\begin{equation}
\sum_{\mu\nu}a_{\mu\nu}=i\,\sigma_{eg}(0),
\label{eq:sum-rule}
\end{equation}
When $\sigma_{eg}(0)\neq0$, the nonzero residue sum pins the numerator degree to $r-1$, one less than this reduced denominator. Thus degree $23$ follows only in a case where the additional $r=24$ conditions have been checked; it is not implied by the overlap count alone. If the initial coherence vanishes, the numerator degree can be lower. Equation~\eqref{eq:sum-rule} is therefore an extraction-independent self-consistency check: any partial-fraction reconstruction with the wrong residue sum has a sign error or a missing term, regardless of the root-finder used.
For the Fock-diagonal cases studied, the numerator polynomial has real coefficients, so its zeros are real or occur in conjugate pairs. Neither real coefficients nor the observed reflection symmetry alone proves that every zero is real. The real-rootedness seen in the tested higher truncations is therefore retained as high-precision numerical evidence. Fock-coherent preparations in the displayed $N_{\max}=2$ family numerically exhibit an upper-half-plane pair. A general symbolic classification of these higher-truncation zeros remains open.

The displayed $N_{\max}=2$ high-precision calculation shows that bath Fock coherence can move coherence zeros into the upper half-plane for that preparation. The tested Fock-diagonal cases remain real-rooted, but a proof for arbitrary truncation, occupation, coupling, and detuning is open. Their force-fit status remains the conditional one stated after Eq.~\eqref{eq:Keff-pert}. Scalar Volterra deconvolution is strongly dependent on normalization, frequency window, and Laplace shift, so no directional claim about correlation-enhanced KK residuals is made; the operator-valued shifted-line calculation in Section~\ref{sec:boundaries} is retained only as a finite-truncation consistency test.

\subsubsection*{Modified KK relation}

We now derive the residue correction in a scalar channel. For this whole-line boundary calculation, strengthen the preceding domain hypothesis to $D=\HH$. Define
\[
 \mathcal H[\phi](\omega)=\frac1\pi\operatorname{p.v.}\!\int_{\mathbb R}
 \frac{\phi(\omega')}{\omega'-\omega}\,d\omega',
\]
so that Eq.~\eqref{eq:kk} reads $\operatorname{Re}f=\mathcal H[\operatorname{Im}f]$ for an upper-half-plane Hardy boundary value. Assume that the exact kernel $\Ktilde$ has the $H^p$ boundary control of Corollary~\ref{cor:subtracted} for some $1<p<\infty$. Assume also that $\tilde{\II}^{(1)}/\tilde\sigma_0$ has only finitely many upper-half-plane poles, all simple and located at zeros $\zeta_j$ of $\tilde\sigma_0$, with $\tilde{\II}^{(1)}(\zeta_j)\neq0$. There are no other upper-half-plane singularities, and after subtracting all these principal parts the quotient remainder has the same $H^p$ boundary control. Finally, assume that Eq.~\eqref{eq:Keff-pert} holds locally uniformly on compact subsets away from the poles and to order $O(\epsilon^2)$ in the boundary norm used below. At a simple zero,
\[
 \frac{\tilde{\II}^{(1)}(z)}{\tilde\sigma_0(z)}
 =\frac{R_j}{z-\zeta_j}+h_j(z),
 \qquad
 R_j=\frac{\tilde{\II}^{(1)}(\zeta_j)}{\tilde\sigma_0'(\zeta_j)},
\]
where $h_j$ is holomorphic near $\zeta_j$. After summing the finitely many principal parts, write
\[
 \Ktilde_{\mathrm{eff}}(z;\epsilon)
 =H_\epsilon(z)+\epsilon\sum_j\frac{R_j}{z-\zeta_j}+O(\epsilon^2),
\]
with $H_\epsilon$ holomorphic in $\HH$ and satisfying the standard boundary relation. The sign can be fixed without convention guessing. For a function $g$ analytic in the upper half-plane and decaying sufficiently on the upper semicircle, the counterclockwise contour gives
\[
 \operatorname{p.v.}\!\int_{\mathbb R}\frac{g(t)}{t-\omega}\,dt=i\pi g(\omega),
 \qquad
 \operatorname{Re}g=\mathcal H[\operatorname{Im}g].
\]
For a function analytic in the lower half-plane with the same decay, the lower contour is clockwise and gives
\[
 \operatorname{p.v.}\!\int_{\mathbb R}\frac{g(t)}{t-\omega}\,dt=-i\pi g(\omega),
 \qquad
 \operatorname{Re}g=-\mathcal H[\operatorname{Im}g].
\]
For $g_j(z)=R_j/(z-\zeta_j)$, the pole lies in $\HH$, so $g_j$ is holomorphic in the lower half-plane and decays as $1/z$. The second contour identity therefore gives
\[
 \operatorname{Re}g_j=-\mathcal H[\operatorname{Im}g_j].
\]
Subtracting the Hilbert transform of the imaginary part from the real part therefore gives two copies of the pole boundary value. Hence
\begin{equation}
\operatorname{Re}[\Ktilde_\mathrm{eff}(\omega)] = \mathcal{H}\big[\operatorname{Im}\Ktilde_\mathrm{eff}\big](\omega) + \epsilon\,\Delta(\omega;\{\zeta_j\}) + O(\epsilon^2),
\label{eq:modified-kk}
\end{equation}
with the explicit correction
\[
\Delta(\omega;\{\zeta_j\}) = 2\sum_j \operatorname{Re}\!\left[\frac{\tilde{\II}^{(1)}(\zeta_j)/\tilde{\sigma}_0'(\zeta_j)}{\omega - \zeta_j}\right].
\]
Here $O(\epsilon^2)$ is not a uniform statement through the pole neighbourhoods; it is understood on compact boundary sets separated from the projected zero locations, or in the stated boundary norm after the principal parts have been removed.

Writing $R_j=\tilde{\II}^{(1)}(\zeta_j)/\tilde{\sigma}_0'(\zeta_j)$, the individual profile is
\[
2\operatorname{Re}\!\left[\frac{R_j}{\omega-\zeta_j}\right]
=\frac{2\{\operatorname{Re}R_j\,[\omega-\operatorname{Re}\zeta_j]
-\operatorname{Im}R_j\,\operatorname{Im}\zeta_j\}}
{[\omega-\operatorname{Re}\zeta_j]^2+[\operatorname{Im}\zeta_j]^2}.
\]
Thus a general complex residue produces a dispersive odd component together with a Lorentzian even component.  The profile is purely Lorentzian only when $\operatorname{Re}R_j=0$; its overall amplitude and phase set the channel-resolved force-fit KK residual.

\section{Rational Pole Obstruction and CPTP Compatibility}
\label{sec:consequences}

\subsection{CPTP obstruction from uncancelled rational poles}

\begin{theorem}[CPTP obstruction from uncancelled rational poles]
\label{thm:cp-hardy}
Let the system Liouville space be finite-dimensional and set $d=\dim\mathcal{B}_{\mathrm{HS}}(\mathcal H_s)$. Let $\Ktilde(z)$ denote a rational kernel reconstruction (for example, a Pad\'e approximant) with an isolated simple pole at $z_0 \in \HH$, $\operatorname{Im}(z_0) > 0$, and residue matrix $R$. The algebraic conclusions below concern this rational reconstruction; they do not by themselves assert that the exact NZ transform has the same pole.
Write $\Ktilde(z) = R/(z-z_0) + \Ktilde_\mathrm{reg}(z)$, with $\Ktilde_\mathrm{reg}$ analytic near $z_0$, and assume $z_0$ is not an eigenvalue of $\LL_s + \Ktilde_\mathrm{reg}(z_0)$.
On the non-pole domain, define the associated propagator transform by
\[
  \tilde{\GG}(z)=i\bigl(z-\LL_s-\Ktilde(z)\bigr)^{-1}.
\]
Choose $r \in (0,\operatorname{Im}(z_0))$ such that $z-\LL_s-\Ktilde_\mathrm{reg}(z)$ is invertible throughout the closed disc $\overline{D}(z_0,r)$, and define
\begin{equation}
n(r) \;\equiv\; \min_{|z-z_0|=r} \bigl\|(z-\LL_s-\Ktilde_\mathrm{reg}(z))^{-1}\bigr\|_\opnorm^{-1} \;>\; 0 .
\label{eq:n-of-r}
\end{equation}
\emph{(i) Local bound (Rouch\'e).}
If $\|R\| < r\,n(r)$ for some $r \in (0,\operatorname{Im}(z_0))$, then the locally cleared analytic determinant has exactly $d$ zeros, counted with multiplicity, in $D(z_0,r)\subset\HH$.
If at least one counted zero $\xi\neq z_0$ satisfies $\operatorname{adj}[A(\xi)-R]\neq0$, where $A(z)=(z-z_0)(z-\LL_s-\Ktilde_\mathrm{reg}(z))$, then $\xi$ is an upper-half-plane propagator pole.

\emph{(ii) Global Vieta criterion under explicit no-cancellation hypotheses.}
Write the proper regular part as $\Ktilde_{\mathrm{reg}}(z)=\mathsf P(z)/q(z)$, where $q$ is monic of degree $m\geq0$, $\mathsf P\equiv0$ if $m=0$, and $\deg \mathsf P\leq m-1$ if $m\geq1$. For $m\geq1$, let $q_{m-1}\in\mathbb R$ be the coefficient of $z^{m-1}$; for $m=0$, set $q_{m-1}:=0$. Define
\[
 M(z)=(z-z_0)\bigl(zq(z)-q(z)\LL_s-\mathsf P(z)\bigr)-R\,q(z).
\]
Assume $q(z_0)\neq0$ and $\det R\neq0$; assume $\det \mathsf P(\alpha)\neq0$ for every root $\alpha$ of $q$; and assume $\operatorname{adj}M(z)\neq0$ at every zero $z\in\HH$ of $\det M$. Then $\tilde{\GG}(z)$ has at least one pole in $\HH$.

Under either part, any resulting genuine pole is incompatible with the associated rational resolvent being the Fourier--Laplace transform of a strongly measurable, uniformly bounded CPTP reduced family on $t\geq0$, provided the asserted rational representation agrees with that transform on a nonempty open subset of $\HH$ away from the rational poles. Thus the CPTP conclusion applies only when the reconstruction is an exact representation on such an open set, not when it is merely an approximation. Norm-continuous microscopic reduced families satisfy the measurability requirement automatically.
\end{theorem}

\begin{proof}
\emph{Finite dimension.}
With $\dim\mathcal{H}_s < \infty$, the effective Liouvillian $\LL_s + \Ktilde(z)$ acts on a finite-dimensional Liouville space; poles of $\tilde{\GG}(z) = i(z-\LL_s-\Ktilde(z))^{-1}$ are tracked through zeros of $D(z) \equiv \det(z-\LL_s-\Ktilde(z))$.
Write $A(z) = (z-z_0)\bigl(z-\LL_s-\Ktilde_\mathrm{reg}(z)\bigr)$, so that $(z-z_0)\bigl[z-\LL_s-\Ktilde(z)\bigr] = A(z)-R$.
For the local argument set $G_{\mathrm{loc}}(z)=\det(A(z)-R)=(z-z_0)^dD(z)$. This is analytic on a neighbourhood of the closed disc. Its zeros away from $z_0$ are denominator zeros of the propagator; they are genuine poles only after the numerator/adjugate check.

\emph{Part (i): Rouch\'e bound on the circle $C_r = \{|z-z_0|=r\}$.}
For every $z \in C_r$ the Neumann series gives
\[
\bigl\|\bigl[A(z)-R\bigr]^{-1}\bigr\|_\opnorm
\;=\; \bigl\|A(z)^{-1}\bigl(I-RA(z)^{-1}\bigr)^{-1}\bigr\|_\opnorm
\;\leq\; \frac{\|A(z)^{-1}\|_\opnorm}{1-\|R\|\,\|A(z)^{-1}\|_\opnorm}
\]
provided $\|R\| < \|A(z)^{-1}\|_\opnorm^{-1}$.
Now
\[
\|A(z)^{-1}\|_\opnorm^{-1}
= |z-z_0|\,\bigl\|(z-\LL_s-\Ktilde_\mathrm{reg}(z))^{-1}\bigr\|_\opnorm^{-1}
= r\,\bigl\|(z-\LL_s-\Ktilde_\mathrm{reg}(z))^{-1}\bigr\|_\opnorm^{-1}
\geq r\,n(r) .
\]
Hence, if $\|R\| < r\,n(r)$, then $\|R\| < \|A(z)^{-1}\|_\opnorm^{-1}$ for every $z \in C_r$, and consequently $A(z)-R$ is invertible on $C_r$; i.e.\ $G_{\mathrm{loc}}(z)\neq 0$ on $C_r$.
Because $G_{\mathrm{loc}}(z) = \det(A(z)-R)$ depends continuously on $R$, the map $s \mapsto G_{\mathrm{loc},s}(z) = \det(A(z)-sR)$ is a homotopy with no zeros on $C_r$ for $s\in[0,1]$.
By the argument principle, $G_{\mathrm{loc},1}(z)=G_{\mathrm{loc}}(z)$ and $G_{\mathrm{loc},0}(z)=\det(A(z))$ have the same number of zeros inside $D(z_0,r)$.
But $G_{\mathrm{loc},0}(z) = (z-z_0)^{d}\det(z-\LL_s-\Ktilde_\mathrm{reg}(z))$ with $d = \dim\mathcal{B}_\mathrm{HS}(\mathcal{H}_s)$.
By the closed-disc regularity assumption, $\det(z-\LL_s-\Ktilde_\mathrm{reg}(z))\neq 0$ on $\overline{D}(z_0,r)$; therefore $G_{\mathrm{loc},0}$ has exactly $d$ zeros inside $D(z_0,r)$, all at $z_0$.
Thus $G_{\mathrm{loc}}(z)$ also has exactly $d$ zeros inside $D(z_0,r) \subset \HH$.
For a counted zero $\xi\neq z_0$, the identity
\[
 \tilde{\GG}(z)=i(z-z_0)\frac{\operatorname{adj}[A(z)-R]}{\det[A(z)-R]}
\]
shows that $\operatorname{adj}[A(\xi)-R]\neq0$ is a sufficient no-cancellation test for a genuine propagator pole.

\emph{Part (ii): Global existence via Vieta's formula.}
Use the reduced proper representation from part~(ii), $\Ktilde_\mathrm{reg}(z)=\mathsf P(z)/q(z)$ with $\deg q=m$, $q$ monic, and $\deg \mathsf P\leq m-1$.
Multiplying the locally cleared determinant by $q(z)^{d}$ eliminates the remaining denominator:
\[
q(z)^{d}G_{\mathrm{loc}}(z)=[(z-z_0)q(z)]^dD(z)
= \det\!\bigl[(z-z_0)\bigl(zq(z)-q(z)\LL_s-\mathsf P(z)\bigr)-R\,q(z)\bigr] \equiv \det M(z) .
\]
$M(z)$ is a matrix polynomial of degree $m+2$; its determinant is a scalar polynomial of degree $d(m+2)$.
Expanding the leading terms,
\[
M(z) = z^{m+2}I - (z_0+\LL_s-q_{m-1}I)z^{m+1} + O(z^{m}) ,
\]
where $q_{m-1}$ has the convention stated in part~(ii), including $q_{m-1}=0$ for $m=0$.
Vieta's formula gives the sum of \emph{all} $d(m+2)$ zeros of $\det M(z)$:
\begin{equation}
\sum_{j=1}^{d(m+2)} z_j = \operatorname{Tr}(z_0+\LL_s-q_{m-1}I) = d\,z_0 + \operatorname{Tr}(\LL_s) - d\,q_{m-1} .
\label{eq:vieta-all}
\end{equation}
At a root $\alpha$ of $q$ the first term of $M$ vanishes and $M(\alpha) = -(\alpha - z_0)\mathsf P(\alpha)$.
The hypotheses $q(z_0)\neq0$ and $\det \mathsf P(\alpha)\neq0$ therefore make $M(\alpha)$ invertible at every root of $q$. At $z=z_0$, $M(z_0)=-R\,q(z_0)$ is invertible because $\det R\neq0$. Thus neither the pole location $z_0$ nor a root of $q$ is a zero introduced by clearing denominators.

Two facts close the coefficient calculation. First, $\operatorname{Tr}(\LL_s) = 0$ identically.
In the Hilbert-Schmidt basis $\{E_{ij} = |i\rangle\langle j|\}$, $\langle\!\langle E_{ij},\LL_s E_{ij}\rangle\!\rangle_\mathrm{HS} = \langle i|H_s|i\rangle - \langle j|H_s|j\rangle$; summing over $i,j$ gives $0$ independent of $H_s$.
Second, $q_{m-1}\in\mathbb R$ by the explicit hypothesis of part~(ii).
\emph{The imaginary-part budget.} Taking imaginary parts of Eq.~\eqref{eq:vieta-all},
\begin{equation}
\sum_{j=1}^{d(m+2)} \operatorname{Im}(z_j) \;=\; d\,\operatorname{Im}(z_0) \;>\; 0 ,
\label{eq:vieta-propagator}
\end{equation}
so at least one zero of $\det M$ lies in $\HH$. The no-clearing-artifact and no-adjugate-cancellation hypotheses convert at least one such zero into a genuine propagator pole.
The inequality is purely algebraic: a fixed positive amount of imaginary part, $d\operatorname{Im}(z_0)$, is distributed among the $d(m+2)$ zeros of the pencil, and no distribution can place all of them in the closed lower half-plane.
Indeed, for such a zero $z_\ast$ the scalar factor $(z_\ast-z_0)q(z_\ast)$ is nonzero and $\operatorname{adj}M(z_\ast)\neq0$. Hence at least one entry of
\[
 \tilde{\GG}(z)=i(z-z_0)q(z)\,\frac{\operatorname{adj}M(z)}{\det M(z)}
\]
has a non-removable pole at $z_\ast$. If $R$ is singular, $\mathsf P(\alpha)$ is singular at a denominator root, or the adjugate vanishes at every UHP zero, the Vieta identity still locates polynomial zeros but does not by itself prove a propagator pole; these cases are deliberately outside part~(ii).

\emph{CPTP contradiction.}
If $\{\Lambda_t\}_{t\geq0}$ is strongly measurable and uniformly bounded in finite dimension, then each scalar channel $t\mapsto[\Lambda_t]_{ab}$ is measurable and bounded. Exponential damping makes its one-sided Fourier--Laplace integral holomorphic throughout $\HH$, with a bound proportional to $1/\operatorname{Im}z$. If the rational resolvent is claimed to represent this reduced family, the two transforms agree by construction on a nonempty open subset of $\HH$ where the resolvent formula is defined (for example, a right half-plane of convergence). The identity theorem then makes them equal on the connected common domain. The bounded transform extends holomorphically across any putative UHP pole, whereas the rational expression has a non-removable singularity there, which is impossible. Every CPTP map is trace-norm contractive, $\|\Lambda_t\|_{1\to1}=1$, so an uncancelled UHP pole is incompatible with a strongly measurable CPTP reduced family. This transform-level argument avoids any assumption about termwise Bromwich inversion.
\end{proof}

\begin{corollary}[Explicit local residue bound]
\label{cor:cp-hardy-explicit}
Set $N \equiv \|(z_0-\LL_s-\Ktilde_\mathrm{reg}(z_0))^{-1}\|_\opnorm$ and $r=\operatorname{Im}(z_0)/2$. If
\[
 n(r)\geq\frac{1}{2N},\qquad \|R\|<\frac{\operatorname{Im}(z_0)}{4N},
\]
then $\|R\|<r\,n(r)$ and conclusion (i) of Theorem~\ref{thm:cp-hardy} holds.
\end{corollary}

\begin{proof}
With $r=\operatorname{Im}(z_0)/2$, the two assumed bounds give $r\,n(r)\geq\operatorname{Im}(z_0)/(4N)>\|R\|$. Theorem~\ref{thm:cp-hardy}(i) applies.
\end{proof}

\begin{remark}[Local vs.\ global information]
\label{rem:local-global}
Part~(i) gives stronger \emph{local} information: it places exactly $d$ denominator zeros inside a prescribed disc $D(z_0,r)$, but requires the small-residue hypothesis $\|R\| < r\,n(r)$ and a no-cancellation check before those zeros are called propagator poles.
Part~(ii) is weaker (only one pole in $\HH$, with no control on its location) and uses stronger global non-cancellation assumptions at $z_0$, the roots of $q$, and the UHP zeros of the polynomial pencil.
For a correlation expansion the residue may be $O(\epsilon)$, but the local theorem applies only after $n(r)$, the residue bound, and the no-cancellation condition have been checked for the reconstruction at hand.
The global result is an algebraic existence statement; a concrete rational reconstruction still requires each explicit non-cancellation check in the theorem.
\end{remark}

\begin{remark}[CPTP vs.\ signal-causality]
An upper-half-plane kernel pole that, after the theorem's stated clearing, denominator, and adjugate no-cancellation tests, becomes a genuine propagator pole in the rational reconstruction is incompatible with the boundedness required of a CPTP reduced family, provided that reconstruction agrees with the reduced-family transform on the theorem's nonempty open set; this is not the signal-propagation sense of causality (which would require super-luminal communication).
The bridge to macroscopic causality is this: CPTP is a necessary consequence of factorized unitary system--bath dynamics under partial trace, so a reconstructed homogeneous resolvent with an uncancelled UHP pole cannot coincide with that exact CPTP reduced family on the open set required in the theorem.
Theorem~\ref{thm:cp-hardy} thus reads such a genuine propagator pole as incompatibility with a CPTP reduction of unitary system--bath dynamics, not as a local super-luminal effect.
\end{remark}

\section{Further Analytic Consequences}
\label{sec:further}

\subsection{Positive Projected Weights and the Anti-Herglotz Sign}

\begin{proposition}[Positive projected weight and the anti-Herglotz sign]
\label{thm:passivity}
Assume the spectral hypotheses of Theorem~\ref{thm:hardy} and suppose the coupling-weighted density is operator-positive, $w(\lambda)\geq0$ as a bounded operator on $\mathcal{B}_{\mathrm{HS}}(\mathcal H_s)$ for almost every $\lambda$. Then $\Ktilde\in\Hp{p}(\BB)$ for every $p\in(1,p_0]$ and obeys the corresponding KK relations. Moreover,
\begin{equation}
 \operatorname{Im}\langle\!\langle\xi,\Ktilde(z)\xi\rangle\!\rangle\leq0,
 \qquad z\in\HH,\quad \xi\in\mathcal{B}_{\mathrm{HS}}(\mathcal H_s),
\label{eq:passivity}
\end{equation}
so the anti-Hermitian operator part $\operatorname{Im}_{\mathrm{op}}\Ktilde=(\Ktilde-\Ktilde^\dagger)/(2i)$ is negative semidefinite. Thus $-\Ktilde$ is operator-valued Herglotz. This operator imaginary part is distinct from the entrywise quantity $\operatorname{Im}_{\mathcal E}$ used in the KK formulas.
\end{proposition}

\begin{proof}
Under the spectral representation~\eqref{eq:spectral-rep}, a positive coupling-weighted density gives the scalar check
\begin{equation}
\operatorname{Im}\langle\!\langle \xi, \Ktilde(z) \xi \rangle\!\rangle = -\int_{-\infty}^{\infty} \frac{y}{(\lambda - x)^2 + y^2}\, w_\xi(\lambda)\,d\lambda \leq 0,
\end{equation}
where $w_\xi(\lambda)=\langle\!\langle \xi,w(\lambda)\xi\rangle\!\rangle\geq0$ and $z=x+iy$.
Since the inequality holds for every $\xi$, it is equivalent to $\operatorname{Im}_{\mathrm{op}}\Ktilde(z)\leq0$. The Hardy membership and KK relations follow independently from Theorem~\ref{thm:hardy}; positivity supplies the Nevanlinna sign structure, not a replacement for the $L^p$ spectral hypothesis.
\end{proof}

Thus positivity adds the anti-Herglotz sign to the independently assumed Hardy placement. The proposition does not claim that microscopic passivity alone makes the projected weight positive or supplies its global $L^p$ control.

\subsection{Determinacy for Bounded Finite Truncations}

The tensorial MKCT equation~\cite{liu2026mkct} for a basis-correlation matrix is
\begin{equation}
\frac{d}{dt}\bm{\mathcal{C}}(t) = \bm{\Omega}_1 \bm{\mathcal{C}}(t) + \int_0^t d\tau\, \bm{\KK}(\tau) \bm{\mathcal{C}}(t-\tau),
\label{eq:mkct}
\end{equation}
with moments $\bm{\Omega}_n = [((i\LL)^n \hat{\phi}_i, \hat{\phi}_j)]_{ij}$.

\begin{theorem}[Moment determinacy for bounded finite truncations]
\label{thm:carleman}
Consider a finite total Hilbert-space truncation in which the projected moment sequence is generated by a bounded self-adjoint spectral problem with a positive, compactly supported matrix-valued measure on the real axis. The measure may be atomic; no absolute-continuity or $L^p$ hypothesis is imposed here.
Because the MKCT convention in Eq.~\eqref{eq:mkct} uses $(i\LL)^n$, the positive Hamburger moments are the sign-corrected even moments
\[
\bm{M}_{2n}=(-1)^n\bm{\Omega}_{2n},
\]
which satisfy $\bm{M}_{2n}\geq0$ on the relevant Liouville-space sector, so the diagonal scalar moments $\langle\!\langle\xi,\bm{M}_{2n}\xi\rangle\!\rangle$ are moments of positive scalar Hamburger measures.
With the convention that a term with $\|\bm\Omega_{2n}\|_\opnorm=0$ is $+\infty$, the Carleman criterion
\begin{equation}
\sum_{n=1}^{\infty} \|\bm{\Omega}_{2n}\|_\opnorm^{-1/(2n)} = \infty
\label{eq:carleman}
\end{equation}
holds, the diagonal moment problems are determinate, and the full finite-dimensional matrix-valued spectral measure is unique.
This is a determinacy diagnostic for Pad\'e or MKCT reconstructions.
It does not by itself prove Hardy membership or convergence of every finite-order Pad\'e sequence; those properties still require the $L^p$ spectral-density assumptions of Theorem~\ref{thm:hardy} or a separate Markov/Stieltjes Pad\'e convergence theorem for the particular model.
\end{theorem}

\begin{proof}
If the supporting interval is contained in $[-R,R]$, finite-dimensional norm equivalence gives $\|\bm\Omega_{2n}\|\leq C R^{2n}$ with $C$ independent of $n$. A zero even moment gives an infinite Carleman term; otherwise $\|\bm\Omega_{2n}\|^{-1/(2n)}\geq C^{-1/(2n)}R^{-1}$, whose positive lower limit forces divergence. Compact support makes each diagonal scalar moment problem determinate by polynomial density, and complex polarization then fixes the full finite-dimensional matrix-valued measure~\cite{akhiezer1965}. This argument uses bounded finite-dimensional propagation and does not extend automatically to an unbounded bath.
\end{proof}

\begin{remark}[Asymptotic criterion vs.\ finite moments]
\label{rem:carleman-scope}
Condition~\eqref{eq:carleman} is asymptotic, whereas a reconstruction uses finitely many moments. For an unbounded continuum bath one must first identify the positive projected measure, prove finiteness of its moments, and establish the norm comparison needed to pass from scalar to operator moments. The raw Drude--Lorentz density, for example, has a $1/\omega$ ultraviolet tail and does not itself supply all ordinary moments. The formalization records the continuum transfer only as a conditional implication with precisely these open inputs.
\end{remark}

\section{Scope and Numerical Illustrations}
\label{sec:boundaries}

The calculation below is a shifted-line, distributional KK consistency check for an extracted finite Jaynes--Cummings kernel. It tests the matrix reconstruction and boundary-sign convention at fixed truncation and regularization. The separate projected-generator scans in Appendix~\ref{app:qlq-spectrum} test only whether the real-spectrum part of the projected-spectral hypothesis holds on the recorded finite grids.

\subsection{Shifted operator-valued KK consistency in the Jaynes--Cummings model}

We now test the shifted boundary identity for an extracted NZ kernel rather than for a bath correlator.
For the parameters used here ($N_{\max}=10$, $g=0.3$, $\omega_0=\omega_c=1$), direct finite-matrix reconstruction gives a real spectrum for $\Qproj\LL\Qproj$ to numerical precision. This checks only the real-spectrum component of the hypothesis. The finite Fock-space truncation has a discrete singular measure and therefore does not satisfy the absolutely continuous $L^p$ assumption of Theorem~\ref{thm:hardy}. The calculation tests the corresponding shifted, distributional boundary identity.

We first extract the physical convolution coefficient $\KK_{\mathrm{phys}}(t)=-i\KK(t)$ as a $4\times4$ matrix by multi-initial-state Volterra deconvolution, and then return to the dispersive convention through $\KK(t)=i\KK_{\mathrm{phys}}(t)$.
For the four factorized initial states $|e\rangle\langle e|$, $|g\rangle\langle g|$, $|+\rangle\langle+|$, and $|+i\rangle\langle+i|$, each tensored with $|0\rangle\langle0|$, define $v^{(j)}(t)=\operatorname{vec}\sigma^{(j)}(t)$ using column stacking and
\[
 V(t)=\bigl[v^{(1)}(t)\ \cdots\ v^{(4)}(t)\bigr],
 \qquad V_0=V(0).
\]
The matrix $V_0$ has rank four and $2$-norm condition number $3.23$. With $F(t)=\dot V(t)+i\LL_sV(t)$, the physical Volterra equation is
\[
 F(t)=\int_0^t\KK_{\mathrm{phys}}(\tau)V(t-\tau)\,d\tau.
\]
On the uniform grid $t_n=n\Delta t$, the recorded implementation evaluates $\dot V$ by centered finite differences in the interior and one-sided differences at the endpoints, and uses the rectangular convolution rule
\[
 F_n=\Delta t\sum_{m=0}^{n}\KK_{\mathrm{phys},m}V_{n-m}.
\]
Thus the step-$n$ value is obtained recursively from
\[
 \KK_{\mathrm{phys}}(t_n)=\left[\frac{F(t_n)}{\Delta t}
 -\sum_{m=0}^{n-1}\KK_{\mathrm{phys}}(t_m)V(t_{n-m})\right]V_0^{-1}.
\]
This definition fixes the derivative, quadrature, vectorization, and matrix ordering used in the extraction.

The infinite-time transform is approximated on the recorded window by
\[
 \Ktilde_{t_{\max}}(\omega+i\varepsilon)
 =\int_0^{t_{\max}}\KK(t)e^{i(\omega+i\varepsilon)t}\,dt,
 \qquad t_{\max}=60,
\]
with $\varepsilon=0.3$ damping the recurrent, non-decaying JC kernel. Thus the calculation below includes finite-window truncation error; it does not claim that the tail from $t_{\max}$ to infinity was computed.
Figure~\ref{fig:matrix-kk} records the time-domain norm, elementwise residuals, operator-norm residual, and coherence-channel comparison. The active coherence channel has a $0.113\%$ relative residual in the displayed $N_t=4000$ run; channels with numerically zero reference weight are reported only as absolute diagnostics.

On the frequency window $\Omega=(0.05,3)$, write $F(\omega)=\operatorname{Re}_{\mathcal E}\Ktilde_{t_{\max}}(\omega+i\varepsilon)$ and $G(\omega)=\mathcal H[\operatorname{Im}_{\mathcal E}\Ktilde_{t_{\max}}(\,\cdot\,+i\varepsilon)](\omega)$. The reported integrated relative residual is
\[
 \mathcal R_{L^2,\mathrm{op}}=\frac{\|F-G\|_{L^2(\Omega;\mathrm{op})}}{\|F\|_{L^2(\Omega;\mathrm{op})}},
 \qquad \|X\|_{L^2(\Omega;\mathrm{op})}^2:=\int_\Omega\|X(\omega)\|_\opnorm^2\,d\omega,
\]
with both norms evaluated by the same uniform FFT-frequency sum. At $N_t=8000$ the finest measured value, $0.024\%$, aggregates contributions from the Volterra inversion, the finite time grid, and the circular FFT-based Hilbert transform used for the KK check. It decreases monotonically under grid refinement for $N_t=2000,4000,8000$,
\[
0.113\%\ \to\ 0.035\%\ \to\ 0.024\%,
\]
with the coherence-channel residual $\Ktilde_{eg,eg}$ decreasing in parallel,
\[
0.459\%\ \to\ 0.113\%\ \to\ 0.028\%.
\]
This grid-refinement trend is measured at fixed $N_{\max}$, $t_{\max}$, frequency window, and Laplace shift. It supports numerical consistency with the shifted distributional identity; it is not a continuum-bath or thermodynamic-limit result.
Within this fixed finite-truncation setup, the extracted $4\times4$ kernel is consistent with the shifted operator-valued KK identity, and its projected generator satisfies the tested real-spectrum condition. The reported relative residuals are not interpreted for channels whose reference weight is numerically zero.

\begin{figure}[!tbp]
\centering
\includegraphics[width=\textwidth]{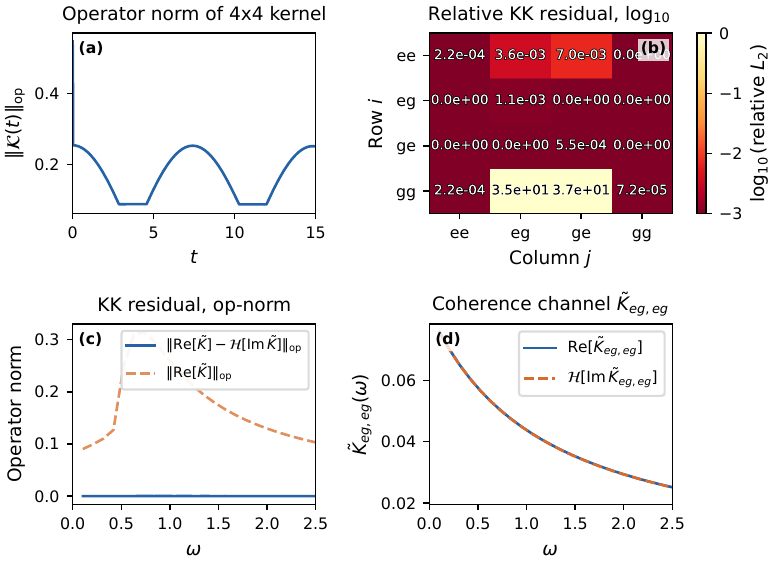}
\caption{Shifted-line operator-valued KK consistency check for the finite JC model ($\omega_0 = \omega_c = 1$, $g = 0.3$, $N_{\max} = 10$, $N_t=4000$, $t_{\max}=60$, $\varepsilon=0.3$).
(a)~Operator norm $\|\KK(t)\|_\mathrm{op}$ of the extracted 4$\times$4 NZ memory kernel.
(b)~Elementwise relative KK residual ($\log_{10}$ scale) across all 16 matrix elements.
(c)~Operator-norm KK residual $\|\operatorname{Re}_{\mathcal E}\Ktilde - \mathcal{H}[\operatorname{Im}_{\mathcal E}\Ktilde]\|_\mathrm{op}$ (blue) vs.\ $\|\operatorname{Re}_{\mathcal E}\Ktilde\|_\mathrm{op}$ (red dashed), with real and imaginary parts taken entrywise in the fixed basis $\mathcal E$.
(d)~Coherence-channel comparison: $\operatorname{Re}[\Ktilde_{eg,eg}(\omega)]$ (blue) and $\mathcal{H}[\operatorname{Im}\,\Ktilde_{eg,eg}](\omega)$ (red dashed).}
\label{fig:matrix-kk}
\end{figure}

\section{Machine-Checked Formalization}
\label{sec:lean-certification}

Formal verification is becoming an essential tool in mathematical physics, offering machine-checked certainty for complex analytic arguments. As formalized libraries grow and AI-assisted proving matures, computer-verified proofs are poised to become standard practice in the field. We therefore formalize the paper-specific deductions in Lean~4~\cite{lean4,mathlib}, providing independent verification of the analytic chain from explicitly named classical inputs.

Table~\ref{tab:lean-status} in Appendix~\ref{app:lean} records the machine-checked boundary of this development. The code verifies the finite-dimensional operator assembly, the Rouch\'e and Vieta deductions, the residue threshold, finite-truncation Carleman determinacy, and the non-pole $N_{\max}=1$ Jaynes--Cummings algebra. The formalization certifies that the stated implications follow from their hypotheses; the physical identification of spectral measures and the numerical calculations remain external inputs. The audit confirms that all proofs are complete and rest only on standard logical foundations.

\section{Discussion and Outlook}
\label{sec:discussion}

The main outcome is an explicit separation of three notions that are often conflated: retarded support, upper-half-plane holomorphy, and Hardy boundary control. Under explicit projected-spectral hypotheses this separation yields the operator-valued KK boundary identity, while the rational-pencil analysis supplies clearing and non-cancellation tests that turn a candidate kernel singularity into a genuine UHP propagator pole. The CPTP obstruction then applies only when that rational resolvent agrees with the exact reduced-family transform on the required nonempty open set. The main Hardy theorem is a sufficient-condition result: strict membership requires an absolutely continuous projected weight in $L^1\cap L^{p_0}$, and the $H^1$ boundary theorem is a separate endpoint statement. These distinctions matter in finite truncations, where the Jaynes--Cummings calculation tests a shifted distributional identity rather than strict $H^p$ placement.

The state and force-fit transforms have different protections. Microscopic unitarity makes $\tilde\sigma$ holomorphic for factorized and correlated initial states, but it does not make the transform zero-free. The high-precision Jaynes--Cummings zeros therefore remain in the state-transform layer; the present numerics do not activate the force-fit or CPTP mechanism. If a genuine rational UHP propagator pole survives the clearing factors, Theorem~\ref{thm:cp-hardy} gives the CPTP obstruction through the local count or the global imaginary-part budget.

The most immediate mathematical extension is to verify the projected spectral hypothesis in nontrivial continuum models rather than infer it from a bath label. For reconstruction theory, the important missing ingredients are stable multi-trajectory operator-kernel extraction and explicit non-cancellation tests. In the continuum Carleman problem, the projected positive measure and the scalar-to-operator norm comparison must be established before determinacy can be claimed. Multi-time NZ kernels have been developed for multitime correlation functions~\cite{breuer2017}. Extending the present result to that setting would require genuinely multivariable boundary theory rather than a formal repetition of the one-frequency argument.

\backmatter

\section*{Declarations}

\begin{itemize}
\item \textbf{Funding.} This work was supported by the National High-Level Overseas Talent Program (KS21400126), the Suzhou Talent project (ZXP2025057), the Jiangsu Distinguished Professorship Fund (SR21400225), and the Research Start-up Fund (NH21400525).
\item \textbf{Conflict of interest.} The author declares no competing interests.
\item \textbf{Data availability.} The numerical code, data, and Lean 4 sources used for this version are included in the supplementary archive. A numerical reproduction archive is available at Zenodo, \url{https://doi.org/10.5281/zenodo.20773898}.
\item \textbf{Code availability.} Complete audited Lean 4 sources for this version are included in the supplementary archive; see also the public repository~\cite{leanformalizations}.
\item \textbf{Author contribution.} K.L. conceived and performed the research and wrote the manuscript.
\item \textbf{ORCID.} Kejun Liu: \href{https://orcid.org/0000-0003-1547-9280}{0000-0003-1547-9280}.
\end{itemize}

\begin{appendices}
\renewcommand{\theHequation}{\Alph{section}.\arabic{equation}}
\renewcommand{\theHfigure}{\Alph{section}.\arabic{figure}}
\renewcommand{\theHtable}{\Alph{section}.\arabic{table}}

\section{Supplement to Theorem~\ref{thm:cp-hardy}: implicit-function and Vieta viewpoints}
\label{app:cp-hardy}

Theorem~\ref{thm:cp-hardy} gives local and global criteria for propagator poles generated by a UHP kernel pole. Here we record the perturbative expansion underlying the local picture and the separate Vieta viewpoint.

\emph{Perturbative local viewpoint.}
Write $\Ktilde(z) = R/(z-z_0) + \Ktilde_\mathrm{reg}(z)$ and define $F(z) = (z-z_0)(z-\LL_s-\Ktilde_\mathrm{reg}(z))$.
If $A_0 = z_0-\LL_s-\Ktilde_\mathrm{reg}(z_0)$ is invertible, the zeros of the locally cleared determinant are obtained from
$\det[(z-z_0)A_0-R+O((z-z_0)^2)]=0$.
For a one-parameter residue $R=\epsilon R_0$, if $A_0^{-1}R_0$ has simple spectrum, analytic perturbation gives
\begin{equation}
\zeta_j = z_0 + \epsilon\mu_j + O(\epsilon^2),
\qquad
\mu_j\in \operatorname{spec}(A_0^{-1}R_0),
\end{equation}
so all $d$ roots remain in the upper half-plane for sufficiently small $|\epsilon|$ compared with $\operatorname{Im}z_0$. With repeated or defective eigenvalues, fractional-power root splitting can occur; the Rouch\'e count in Part~(i) avoids any such perturbative assumption.

\emph{Vieta (global) viewpoint.}
For the rational kernel $\Ktilde(z)=R/(z-z_0)+\mathsf P(z)/q(z)$, the propagator poles satisfy $\det M(z)=0$ with $M(z)=(z-z_0)(zq(z)-q(z)\LL_s-\mathsf P(z))-R\,q(z)$.
The sum of all zeros of the scalar polynomial $\det M(z)$ equals $d(z_0-q_{m-1})+\operatorname{Tr}(\LL_s)$, with $q_{m-1}=0$ when $m=0$.
Under the explicit conditions $q(z_0)\neq0$, $\det R\neq0$, $\det \mathsf P(\alpha)\neq0$ at every root $\alpha$ of $q$, and $\operatorname{adj}M(\xi)\neq0$ at every UHP zero $\xi$ of $\det M$, the positive contribution $d\,\operatorname{Im}z_0$ forces at least one propagator pole in $\HH$.

\emph{No-cancellation anchor on a two-dimensional invariant block ($d=2$, $m=0$).}
To see that these hypotheses are non-vacuous, restrict the pencil to a two-dimensional invariant Liouville block; this is not the full four-dimensional Liouville space of a qubit. Take $\LL_s=\operatorname{diag}(\omega,-\omega)$ with $\omega>0$, set $\Ktilde_\mathrm{reg}=0$ (so $q\equiv1$ and $m=0$), and choose $R=\varrho I_2$ with $\varrho>0$ at $z_0=i\gamma$, $\gamma>0$.
Then
\[
M(z)=(z-i\gamma)(zI-\LL_s)-\varrho I,
\qquad
\det M(z)=m_+(z)m_-(z),
\]
where $m_\pm(z)=(z-i\gamma)(z\mp\omega)-\varrho$. The four zeros are
\[
z_{s}^{\pm}=\tfrac12\left(s\omega+i\gamma
 \pm\sqrt{(s\omega-i\gamma)^2+4\varrho}\right),
\qquad s\in\{+1,-1\}.
\]
The discriminants are nonzero because their imaginary parts are $\mp2\omega\gamma$, so the two roots of each factor are simple. The factors have no common zero: $m_-(z)-m_+(z)=2\omega(z-i\gamma)$, while $m_\pm(i\gamma)=-\varrho\neq0$. Hence all four zeros are distinct. At a zero of one factor the other factor is nonzero, so $\operatorname{adj}M$ is nonzero; moreover no zero equals $z_0$. Thus the scalar clearing factor $(z-z_0)$ and the adjugate cancel none of these zeros. Each pair has imaginary-part sum $\gamma$, and therefore at least one zero from each factor lies in $\HH$. This explicit block example verifies the no-cancellation hypotheses directly; nearby pencils must be checked on their own terms.

\section{Machine-checked status of the main-text statements}
\label{app:lean}

The pinned environment is Lean \texttt{v4.32.0}, mathlib commit \texttt{81a5d257c8e410db227a6665ed08f64fea08e997}, and the isolated Carleson source snapshot \texttt{c82b9f685d}. In the public source-tree layout, run \texttt{lake build Hpaper} and \texttt{lake build Hpaper.AxiomAudit} from the directory containing \texttt{lakefile.toml}.

The formalization was developed through a human-AI collaborative workflow: the core proof architecture and mathematical content were written by the author, with the Kimi K3 LLM model used for debugging, tactic suggestion, and audit assistance. All proofs were subsequently reviewed and verified by the author, who takes full responsibility for the correctness and completeness of the formalization. The audit confirms that all proofs are complete and rest only on standard logical foundations.

\subsection{Formalization architecture}

The Lean development is organized as a hierarchy of modules, each corresponding to a section of the manuscript. The root module \path{Hpaper.lean} imports all submodules and exposes the public interface. The dependency structure ensures that foundational operator-class definitions precede the analytic theorems that consume them.

\begin{itemize}
\item \path{Hpaper.NZDynamics}: Nakajima--Zwanzig projection, GQME derivation, and the explicit operator-class data structure. This module defines the interface between the physical setup and the formal verification, encoding the partial trace, reference-state embedding, and evolution-family hypotheses as structure fields rather than axioms.

\item \path{Hpaper.CausalityTransfer}: Retarded support, spectral decay via Riemann--Lebesgue, and the Cauchy-transform representation. The proof transport from scalar to operator-valued settings is handled here through finite-dimensional componentwise lifting.

\item \path{Hpaper.KKTerminalAssemblyL2}, \path{Hpaper.HardyHalfPlanePoissonLp}, and \path{Hpaper.HardyHalfPlaneFiniteMatrix}: The $L^1\cap L^{p_0}\Rightarrow H^p$ theorem. These modules contain the scalar Cauchy-transform estimate as a named classical input and its finite-dimensional assembly through basis expansion and norm equivalence.

\item \path{Hpaper.KKOperatorLiftCited}, \path{Hpaper.KKAssemblyL2}, and \path{Hpaper.KKTerminalAssemblyLp}: Operator-valued Kramers--Kronig boundary identities, including the $H^1$ endpoint and the $L^p$/subtracted variants. The componentwise lift from scalar Hardy projection theorems is verified explicitly.

\item \path{Hpaper.CPHardyActualPoleObstruction}, \path{Hpaper.CPHardyRationalPencil}, and \path{Hpaper.CPHardyNoCancellation}: The CPTP obstruction theorem, including the local Rouch\'e count and the global Vieta budget. The no-cancellation hypotheses are encoded as separate invertibility assumptions, making the logical dependencies explicit.

\item \path{Hpaper.FiniteCarleman}: Moment determinacy for bounded finite truncations, using the compact-support Hamburger criterion as the named classical input.

\item \path{Hpaper.JaynesCummingsNMaxOne}: Closed-form zero placement for the $N_{\max}=1$ baseline, including the removable-factor analysis for degenerate cases.

\item \path{Hpaper.AxiomAudit}: Post-hoc verification that no custom axioms or proof placeholders remain in the audited declarations.
\end{itemize}

\subsection{Proof engineering decisions}

Several design choices distinguish this formalization from a direct transcription of the manuscript proofs:

\begin{enumerate}
\item \textbf{Explicit hypothesis structures}: Rather than assuming unbounded-generator regularity abstractly, the development packages the required differentiability and continuity properties into the \path{OperatorClassData} structure. This makes the logical dependencies auditable: the Hardy theorem does not secretly invoke Stone's theorem, but rather requires the explicit all-vector interface that the user must supply or verify separately.

\item \textbf{Finite-dimensional reduction}: The operator-valued Hardy placement is proved by componentwise reduction to scalar estimates, not by invoking infinite-dimensional singular-integral theory. This matches the manuscript's finite Liouville-space setting and avoids the UMD machinery that would be needed for general Banach spaces.

\item \textbf{Classical inputs as parameters}: Theorems that rely on external results (scalar Cauchy-transform bounds, Riemann--Lebesgue, Hamburger determinacy) take these as named hypotheses rather than citing them axiomatically. This clarifies the evidence boundary: the formalization certifies the deduction from these inputs, not the inputs themselves.

\item \textbf{No-cancellation interfaces}: The CPTP obstruction theorem separates the algebraic zero count from the analytic pole verification. The adjugate non-cancellation check is a separate hypothesis, allowing the formalization to verify the algebraic budget independently of the analytic continuation arguments.
\end{enumerate}

\subsection{Audit results}

The \path{Hpaper.AxiomAudit} module performs a syntactic scan of all audited declarations, checking for \texttt{sorry}, \texttt{admit}, or custom axioms. The audit covers 2022 declarations and confirms that only the standard logical foundations are used. The audit is re-run automatically on every build, ensuring that subsequent modifications cannot introduce unverified assumptions without detection.

{\footnotesize
\begin{longtable}{@{}p{0.13\textwidth}p{0.26\textwidth}p{0.33\textwidth}p{0.07\textwidth}@{}}
\caption{Machine-check boundary for the manuscript. ``Checked'' always means checked from the hypotheses printed in the corresponding row. Evidence classes are P (paper proof), S (cited standard input), C (conditional), N (numerical), and O (open).}
\label{tab:lean-status}\\
\toprule
Manuscript item & Lean status & External input or limitation & Class \\
\midrule
\endfirsthead
\multicolumn{4}{c}%
{{\bfseries \tablename\ \thetable{} -- continued from previous page}} \\
\toprule
Manuscript item & Lean status & External input or limitation & Class \\
\midrule
\endhead
\midrule \multicolumn{4}{r}{{Continued on next page}} \\ \midrule
\endfoot
\bottomrule
\endlastfoot
Prop.~\ref{prop:lti} & Checked by \path{prop_lti_of_operatorClassData}; finite matrices by \path{prop_lti_finiteMatrix}. & The partial-trace/retraction maps, all-vector pointwise derivative laws for the joint and compressed evolution families, compressed strong continuity, bounded coupling blocks, and Duhamel interface are explicit hypotheses. They are automatic in the finite-matrix instance, but are not consequences of Stone's theorem alone. & S+P \\
Prop.~\ref{thm:causality} & Retarded support and spectral consequences checked in \path{CausalityTransfer}. & Riemann--Lebesgue is the named classical input; support refers to the explicit zero-past extension. & S+P \\
Thm.~\ref{thm:hardy} & Full $L^1\cap L^{p_0}\Rightarrow H^p$ range checked by \path{thm_hardy_matrix_L1Lp0_unconditional}. & Scalar Cauchy-transform/Hardy estimate is the named classical input; finite-dimensional Liouville space and $p=1$ exclusion remain explicit. & S+P \\
Thm.~\ref{thm:kk} & Operator lift checked by \path{thm_kk_matrix_H1_of_scalarH1KK}. & Scalar Hardy projection and scalar $H^1$ KK are named classical inputs. & S+P \\
Cor.~\ref{cor:subtracted} & The $L^p$ lift and subtraction chain are checked by \path{thm_kk_matrix_Lp_of_scalarHardyProjection}, \path{SubtractedKK}, and \path{cor_subtracted_Lp_unconditional}. & Scalar Hardy projection and Hilbert-transform/PV theorems are named classical inputs; boundary identities retain their a.e. and principal-value qualifications. & S+P \\
Thm.~\ref{thm:cp-hardy} & Local count, Vieta budget, no-cancellation interfaces, and pole obstruction checked. & Every invertibility, rationality, reality, and no-cancellation hypothesis remains explicit. & P \\
Cor.~\ref{cor:cp-hardy-explicit} & Checked in \path{CPHardyExplicitBound}. & Uses the corrected factor-four threshold $\|R\|<\operatorname{Im}z_0/(4N)$. & P \\
Prop.~\ref{thm:passivity} & Sign and Hardy consequences checked from the stated positivity and spectral hypotheses. & Does not prove positivity of the projected weight for an arbitrary passive bath. & C+P \\
Thm.~\ref{thm:carleman} & Finite bounded case checked in \path{FiniteCarleman}. & Compact-support moment determinacy is the named classical input; no unconditional continuum-bath transfer. & S+P \\
Prop.~\ref{thm:robustness} & Bounded generators checked directly; general case checked from \path{StoneUnitaryGroupData}. & Stone's theorem and infinite-dimensional trace-class interfaces are named classical inputs. & S+P \\
Prop.~\ref{prop:sigma-zeros} & Closed forms and non-pole zero equivalences checked in \path{JaynesCummingsNMaxOne} and \path{SigmaZeroPlacement}. & The printed proof separately removes the $g=0$ and $p_0=0$ cancellation artifacts; higher-truncation roots are numerical. & P+N \\
Modified KK formula & Single-pole residue algebra checked in \path{ModifiedKKResidue} and \path{ForceFitPerturbation}. & The full function-space perturbative transfer is supplied in the manuscript proof. & P+C \\
Numerical and continuum claims & Outside numbered theorem-prover scope. & QLQ spectra, extracted kernels, higher roots, and continuum Carleman inputs retain numerical or conditional status. & N/O \\
\end{longtable}
}

\section{Numerical tests of the projected-generator spectrum}
\label{app:qlq-spectrum}

This appendix reports the numerical procedure used to test whether the projected generator $\Qproj\LL\Qproj$ has purely real spectrum for the standard Nakajima--Zwanzig partial-trace projection.
The tests address only the real-spectrum component of Theorem~\ref{thm:hardy}. They do not establish absolute continuity or the $L^p$ bounds on the projected weight.

\subsection{Matrix representation of the Liouvillian and the NZ projection}

For a finite-dimensional total Hilbert space $\mathcal{H} = \mathcal{H}_s \otimes \mathcal{H}_b$ with $\dim\mathcal{H}_s = d_s$ and $\dim\mathcal{H}_b = d_b$, set $D=d_s d_b$. The Hamiltonian is a $D\times D$ Hermitian matrix.
The Liouvillian superoperator $\LL\rho = [\hat{H},\rho]$ is represented as a $D^2\times D^2$ matrix in the vectorization (column-stacking) basis:
\begin{equation}
\label{eq:liouvillian-matrix}
L_{\text{mat}} = \mathbbm{1}_D \otimes \hat{H} - \hat{H}^{T} \otimes \mathbbm{1}_D .
\end{equation}
Cyclicality of trace guarantees $L_{\text{mat}}^{\dagger} = L_{\text{mat}}$ in the Hilbert--Schmidt inner product.

The standard NZ projection superoperator $\Proj$ acts as $(\Proj\rho)_{\alpha n, \beta m} = \rho_{b,nm}^{\text{ref}} \sum_k \rho_{\alpha k, \beta k}$.
In the same vectorized basis its matrix elements are
\begin{equation}
\label{eq:nz-projection-matrix}
P_{(\alpha n)(\beta m), (\alpha' k)(\beta' k')}
= \delta_{\alpha\alpha'}\delta_{\beta\beta'}\delta_{kk'}\,\rho_{b,nm}^{\text{ref}} ,
\end{equation}
where Greek indices label system states ($1\leq\alpha,\beta\leq d_s$) and Latin indices label bath states ($1\leq n,m,k\leq d_b$).
This projection satisfies $P^2 = P$, but $P \neq P^{\dagger}$ unless $\rho_b^{\text{ref}} \propto \mathbbm{1}_{d_b}$.
Its complement is $Q = \mathbbm{1}_{D^2} - P$.

The projected generator is then the $D^2\times D^2$ matrix
\begin{equation}
\label{eq:qlq-matrix}
\Qproj\LL\Qproj \;\longleftrightarrow\; Q_{\text{mat}} \, L_{\text{mat}} \, Q_{\text{mat}} .
\end{equation}
Because $Q_{\text{mat}} \neq Q_{\text{mat}}^{\dagger}$ in general, this matrix need not be Hermitian even though $L_{\text{mat}}$ is.

\subsection{Algorithm}

The numerical test proceeds as follows.
\begin{enumerate}
\item Build $\hat{H}$ for the chosen model and truncation.
\item Construct $L_{\text{mat}}$ via Eq.~\eqref{eq:liouvillian-matrix}.
\item Choose a bath reference state $\rho_b^{\text{ref}}$ (vacuum, thermal, or other).
\item Construct $P_{\text{mat}}$ via Eq.~\eqref{eq:nz-projection-matrix} and verify $P_{\text{mat}}^2 = P_{\text{mat}}$ to machine precision.
\item Form $Q_{\text{mat}} = \mathbbm{1}_{D^2} - P_{\text{mat}}$ and compute $M = Q_{\text{mat}} L_{\text{mat}} Q_{\text{mat}}$.
\item Report the non-orthogonality diagnostic $\|P_{\text{mat}} - P_{\text{mat}}^{\dagger}\|_F$; it is not a required pass/fail condition, since a maximally mixed reference can make the projection orthogonal.
\item Compute all eigenvalues $\{\lambda_j\}$ of $M$ using standard dense eigensolvers.
\item Report $\max_j |\operatorname{Im}(\lambda_j)|$.
\end{enumerate}
A threshold of $10^{-10}$ is used to distinguish genuine complex eigenvalues from machine-precision noise. The Jaynes--Cummings scan reaches $D=42$ ($D^2=1764$), while the spin--boson scan reaches $D=54$ ($D^2=2916$). For the matrices whose spectra are numerically real, dense eigensolvers give imaginary parts below $10^{-14}$ up to the reported rounding.

\subsection{Jaynes--Cummings model: tested finite grid}

The JC Hamiltonian
\begin{equation}
\hat{H}_{\text{JC}} = \frac{\omega_0}{2}\sigma_z + \omega_c a^{\dagger}a + g\bigl(\sigma_+ a + \sigma_- a^{\dagger}\bigr)
\end{equation}
is truncated to $0 \leq n \leq N_{\max}$ Fock states, giving $D = 2(N_{\max}+1)$.

Table~\ref{tab:qlq-jc} summarizes the results.
The tracked scan contains 54 configurations: 25 resonant vacuum cases and 25 resonant thermal cases on the grid $N_{\max}=3,5,10,15,20$, $g=0.1,0.3,0.5,1.0,2.0$, together with four detuned thermal cases. On this finite grid, the eigenvalues of $\Qproj\LL\Qproj$ are real to numerical precision ($\max |\operatorname{Im}\lambda| \leq 2.4\times10^{-14}$).
At the same time, $\Qproj\LL\Qproj$ is manifestly non-Hermitian: $\|\Qproj\LL\Qproj - (\Qproj\LL\Qproj)^{\dagger}\|_F \sim 10^1$--$10^2$.
The scan does not establish a general real-spectrum theorem. Excitation-number conservation is a plausible source of the observed structure, but that explanation is not needed for the numerical claim.

\begin{table}[htbp]
\centering
\caption{Representative $Q\mathcal LQ$ spectra for the Jaynes--Cummings model with a vacuum reference state. The statement is limited to the tracked finite grid.}
\label{tab:qlq-jc}
\begin{tabular}{ccccc}
\toprule
$N_{\max}$ & $g$ & $D$ & $\|\Qproj\LL\Qproj-(\Qproj\LL\Qproj)^{\dagger}\|_F$ & $\max |\operatorname{Im}\lambda|$ \\
\midrule
3  & 0.1 & 8  & 6.98  & $5.4\times 10^{-26}$ \\
3  & 1.0 & 8  & 10.8  & $4.5\times 10^{-25}$ \\
3  & 2.0 & 8  & 17.9  & $6.3\times 10^{-25}$ \\
5  & 0.1 & 12 & 13.5  & $5.4\times 10^{-26}$ \\
5  & 1.0 & 12 & 18.4  & $7.5\times 10^{-16}$ \\
5  & 2.0 & 12 & 28.5  & $6.3\times 10^{-25}$ \\
10 & 0.1 & 22 & 35.1  & 0 \\
10 & 1.0 & 22 & 46.3  & $7.5\times 10^{-16}$ \\
10 & 2.0 & 22 & 70.0  & $6.5\times 10^{-16}$ \\
20 & 0.1 & 42 & 94.2  & 0 \\
20 & 1.0 & 42 & 111   & $5.3\times 10^{-15}$ \\
20 & 2.0 & 42 & 150   & $1.4\times 10^{-14}$ \\
\bottomrule
\end{tabular}
\end{table}

\subsection{Spin--boson finite-matrix counterexamples}

For the truncated spin--boson Hamiltonian
\begin{equation}
\hat{H}_{\text{SB}} = \frac{\Delta}{2}\sigma_x + \sum_{k=1}^{N_{\text{bath}}} \omega_k a_k^{\dagger}a_k + \sigma_z \sum_{k=1}^{N_{\text{bath}}} g_k (a_k + a_k^{\dagger}),
\end{equation}
the tracked dense scan uses $N_{\mathrm{bath}}=3$, $n_{\max}=2$, $\Delta=1$, nine values of the coupling parameter $\eta$, five temperatures, and three cutoffs, for 135 finite matrices in total. Of these, 107 have non-real eigenvalues above the numerical threshold and 28 are real to numerical precision; the largest observed $|\operatorname{Im}\lambda|$ is $1.06$. The raw scan is archived as the supplementary data file \texttt{qlq\_phase\_scan.json}.

This finite-grid result is sufficient for the logical point: real spectrum of $Q\mathcal LQ$ is not automatic for the physical partial-trace projection. It is not a phase diagram and is not used to infer an operator-valued spectral density in the continuum model.

\subsection{Implications}

These numerical tests support the conditional formulation adopted in the main text.
For the Jaynes--Cummings model used in Fig.~\ref{fig:matrix-kk}, the real-spectrum part of the spectral hypothesis is satisfied in the tested finite truncation; strict Hardy membership still requires the continuous-density assumption of Theorem~\ref{thm:hardy}.
The spin--boson scan supplies explicit finite-matrix cases where the real-spectrum condition fails, so Theorem~\ref{thm:hardy} does not apply there without additional structure.
This model-dependence underscores why the theorems are stated as \emph{sufficient conditions} rather than universal consequences of unitarity alone.

\end{appendices}

\clearpage
{\setlength{\bibsep}{0.6em}\bibliography{references}}


\begin{thebibliography}{17}
\ifx \bisbn   \undefined \def \bisbn  #1{ISBN #1}\fi
\ifx \binits  \undefined \def \binits#1{#1}\fi
\ifx \bauthor  \undefined \def \bauthor#1{#1}\fi
\ifx \batitle  \undefined \def \batitle#1{#1}\fi
\ifx \bjtitle  \undefined \def \bjtitle#1{#1}\fi
\ifx \bvolume  \undefined \def \bvolume#1{\textbf{#1}}\fi
\ifx \byear  \undefined \def \byear#1{#1}\fi
\ifx \bissue  \undefined \def \bissue#1{#1}\fi
\ifx \bfpage  \undefined \def \bfpage#1{#1}\fi
\ifx \blpage  \undefined \def \blpage #1{#1}\fi
\ifx \burl  \undefined \def \burl#1{\textsf{#1}}\fi
\ifx \doiurl  \undefined \def \doiurl#1{\url{https://doi.org/#1}}\fi
\ifx \betal  \undefined \def \betal{\textit{et al.}}\fi
\ifx \binstitute  \undefined \def \binstitute#1{#1}\fi
\ifx \binstitutionaled  \undefined \def \binstitutionaled#1{#1}\fi
\ifx \bctitle  \undefined \def \bctitle#1{#1}\fi
\ifx \beditor  \undefined \def \beditor#1{#1}\fi
\ifx \bpublisher  \undefined \def \bpublisher#1{#1}\fi
\ifx \bbtitle  \undefined \def \bbtitle#1{#1}\fi
\ifx \bedition  \undefined \def \bedition#1{#1}\fi
\ifx \bseriesno  \undefined \def \bseriesno#1{#1}\fi
\ifx \blocation  \undefined \def \blocation#1{#1}\fi
\ifx \bsertitle  \undefined \def \bsertitle#1{#1}\fi
\ifx \bsnm \undefined \def \bsnm#1{#1}\fi
\ifx \bsuffix \undefined \def \bsuffix#1{#1}\fi
\ifx \bparticle \undefined \def \bparticle#1{#1}\fi
\ifx \barticle \undefined \def \barticle#1{#1}\fi
\bibcommenthead
\ifx \bconfdate \undefined \def \bconfdate #1{#1}\fi
\ifx \botherref \undefined \def \botherref #1{#1}\fi
\ifx \url \undefined \def \url#1{\textsf{#1}}\fi
\ifx \bchapter \undefined \def \bchapter#1{#1}\fi
\ifx \bbook \undefined \def \bbook#1{#1}\fi
\ifx \bcomment \undefined \def \bcomment#1{#1}\fi
\ifx \oauthor \undefined \def \oauthor#1{#1}\fi
\ifx \citeauthoryear \undefined \def \citeauthoryear#1{#1}\fi
\ifx \endbibitem  \undefined \def \endbibitem {}\fi
\ifx \bconflocation  \undefined \def \bconflocation#1{#1}\fi
\ifx \arxivurl  \undefined \def \arxivurl#1{\textsf{#1}}\fi
\csname PreBibitemsHook\endcsname

\bibitem[\protect\citeauthoryear{Nakajima}{1958}]{nakajima1958}
\begin{barticle}
\bauthor{\bsnm{Nakajima}, \binits{S.}}:
\batitle{On quantum theory of transport phenomena}.
\bjtitle{Prog. Theor. Phys.}
\bvolume{20},
\bfpage{948}--\blpage{959}
(\byear{1958})
\doiurl{10.1143/PTP.20.948}
\end{barticle}
\endbibitem

\bibitem[\protect\citeauthoryear{Zwanzig}{1960}]{zwanzig1960}
\begin{barticle}
\bauthor{\bsnm{Zwanzig}, \binits{R.}}:
\batitle{Ensemble method in the theory of irreversibility}.
\bjtitle{J. Chem. Phys.}
\bvolume{33},
\bfpage{1338}--\blpage{1341}
(\byear{1960})
\doiurl{10.1063/1.1731409}
\end{barticle}
\endbibitem

\bibitem[\protect\citeauthoryear{Mori}{1965}]{mori1965}
\begin{barticle}
\bauthor{\bsnm{Mori}, \binits{H.}}:
\batitle{Transport, collective motion, and brownian motion}.
\bjtitle{Prog. Theor. Phys.}
\bvolume{33},
\bfpage{423}--\blpage{455}
(\byear{1965})
\doiurl{10.1143/PTP.33.423}
\end{barticle}
\endbibitem

\bibitem[\protect\citeauthoryear{Kubo}{1957}]{kubo1957}
\begin{barticle}
\bauthor{\bsnm{Kubo}, \binits{R.}}:
\batitle{Statistical-mechanical theory of irreversible processes. {I}. general
  theory and simple applications to magnetic and conduction problems}.
\bjtitle{J. Phys. Soc. Jpn.}
\bvolume{12},
\bfpage{570}--\blpage{586}
(\byear{1957})
\doiurl{10.1143/JPSJ.12.570}
\end{barticle}
\endbibitem

\bibitem[\protect\citeauthoryear{Toll}{1956}]{toll1956}
\begin{barticle}
\bauthor{\bsnm{Toll}, \binits{J.S.}}:
\batitle{Causality and the dispersion relation: Logical foundations}.
\bjtitle{Phys. Rev.}
\bvolume{104},
\bfpage{1760}--\blpage{1770}
(\byear{1956})
\doiurl{10.1103/PhysRev.104.1760}
\end{barticle}
\endbibitem

\bibitem[\protect\citeauthoryear{King}{2009}]{king2009}
\begin{bbook}
\bauthor{\bsnm{King}, \binits{F.W.}}:
\bbtitle{Hilbert Transforms}
vol. \bseriesno{1--2}.
\bpublisher{Cambridge University Press},
\blocation{Cambridge}
(\byear{2009})
\end{bbook}
\endbibitem

\bibitem[\protect\citeauthoryear{Prevedelli et~al.}{2024}]{perinelli2024}
\begin{barticle}
\bauthor{\bsnm{Prevedelli}, \binits{M.}},
\bauthor{\bsnm{Perinelli}, \binits{A.}},
\bauthor{\bsnm{Ricci}, \binits{L.}}:
\batitle{{Kramers-Kronig} relations via {Laplace} formalism and $l^1$
  integrability}.
\bjtitle{Am. J. Phys.}
\bvolume{92},
\bfpage{859}--\blpage{863}
(\byear{2024})
\doiurl{10.1119/5.0217609}
\end{barticle}
\endbibitem

\bibitem[\protect\citeauthoryear{Vacchini and Breuer}{2010}]{vacchini2010exact}
\begin{barticle}
\bauthor{\bsnm{Vacchini}, \binits{B.}},
\bauthor{\bsnm{Breuer}, \binits{H.-P.}}:
\batitle{Exact master equations for the non-{Markovian} decay of a qubit}.
\bjtitle{Phys. Rev. A}
\bvolume{81},
\bfpage{042103}
(\byear{2010})
\doiurl{10.1103/PhysRevA.81.042103}
{\href{https://arxiv.org/abs/1002.2172}{{arXiv:1002.2172}}}
\end{barticle}
\endbibitem

\bibitem[\protect\citeauthoryear{Ivander et~al.}{2024}]{ivander2024}
\begin{barticle}
\bauthor{\bsnm{Ivander}, \binits{F.}},
\bauthor{\bsnm{Lindoy}, \binits{L.P.}},
\bauthor{\bsnm{Lee}, \binits{J.}}:
\batitle{Unified framework for open quantum dynamics with memory}.
\bjtitle{Nat. Commun.}
\bvolume{15},
\bfpage{8087}
(\byear{2024})
\doiurl{10.1038/s41467-024-52081-3}
\end{barticle}
\endbibitem

\bibitem[\protect\citeauthoryear{de~Moura and Ullrich}{2021}]{lean4}
\begin{bchapter}
\bauthor{\bsnm{Moura}, \binits{L.}},
\bauthor{\bsnm{Ullrich}, \binits{S.}}:
\bctitle{The {Lean} 4 theorem prover and programming language}.
In: \bbtitle{Automated Deduction -- {CADE} 28}.
\bsertitle{Lecture Notes in Computer Science},
vol. \bseriesno{12699},
pp. \bfpage{625}--\blpage{635}
(\byear{2021})
\end{bchapter}
\endbibitem

\bibitem[\protect\citeauthoryear{{The mathlib Community}}{2020}]{mathlib}
\begin{bchapter}
\bauthor{\bsnm{{The mathlib Community}}}:
\bctitle{The {Lean} mathematical library}.
In: \bbtitle{Proceedings of the 9th {ACM} {SIGPLAN} International Conference on
  Certified Programs and Proofs ({CPP} 2020)},
pp. \bfpage{367}--\blpage{381}
(\byear{2020})
\end{bchapter}
\endbibitem

\bibitem[\protect\citeauthoryear{Diestel and Uhl}{1977}]{diestel1977}
\begin{bbook}
\bauthor{\bsnm{Diestel}, \binits{J.}},
\bauthor{\bsnm{Uhl}, \binits{J.J.}}:
\bbtitle{Vector Measures}.
\bsertitle{Mathematical Surveys and Monographs},
vol. \bseriesno{15}.
\bpublisher{American Mathematical Society},
\blocation{Providence, RI}
(\byear{1977})
\end{bbook}
\endbibitem

\bibitem[\protect\citeauthoryear{Gesztesy and
  Tsekanovskii}{2000}]{gesztesy2000}
\begin{barticle}
\bauthor{\bsnm{Gesztesy}, \binits{F.}},
\bauthor{\bsnm{Tsekanovskii}, \binits{E.}}:
\batitle{On matrix-valued {Herglotz} functions}.
\bjtitle{Math. Nachr.}
\bvolume{218},
\bfpage{61}--\blpage{138}
(\byear{2000})
\doiurl{10.1002/1522-2616(200010)218:1<61::AID-MANA61>3.0.CO;2-D}
\end{barticle}
\endbibitem

\bibitem[\protect\citeauthoryear{Bi et~al.}{2026}]{liu2026mkct}
\begin{botherref}
\oauthor{\bsnm{Bi}, \binits{R.-H.}},
\oauthor{\bsnm{Liu}, \binits{W.}},
\oauthor{\bsnm{Dou}, \binits{W.}}:
Generalized quantum master equation from memory kernel coupling theory.
arXiv preprint
(2026)
{\href{https://arxiv.org/abs/2603.01458}{{arXiv:2603.01458}}}
\end{botherref}
\endbibitem

\bibitem[\protect\citeauthoryear{Akhiezer}{1965}]{akhiezer1965}
\begin{bbook}
\bauthor{\bsnm{Akhiezer}, \binits{N.I.}}:
\bbtitle{The Classical Moment Problem and Some Related Questions in Analysis}.
\bpublisher{Oliver and Boyd},
\blocation{Edinburgh}
(\byear{1965})
\end{bbook}
\endbibitem

\bibitem[\protect\citeauthoryear{Ivanov and Breuer}{2015}]{breuer2017}
\begin{barticle}
\bauthor{\bsnm{Ivanov}, \binits{A.}},
\bauthor{\bsnm{Breuer}, \binits{H.-P.}}:
\batitle{Extension of the {Nakajima-Zwanzig} approach to multitime correlation
  functions}.
\bjtitle{Phys. Rev. A}
\bvolume{92},
\bfpage{032113}
(\byear{2015})
\doiurl{10.1103/PhysRevA.92.032113}
\end{barticle}
\endbibitem

\bibitem[\protect\citeauthoryear{Liu}{2026}]{leanformalizations}
\begin{botherref}
\oauthor{\bsnm{Liu}, \binits{K.}}:
Causality and open quantum systems -- {Lean} 4 formalizations.
\url{https://github.com/KejunLiuGitHub/causality-open-quantum-systems-lean}
(2026)
\end{botherref}
\endbibitem

\end{thebibliography}

\end{document}